\documentclass[11pt]{article}
\usepackage[a4paper,bindingoffset=0.in,left=2.5cm,right=2.5cm,top=2cm,bottom=3cm]{geometry}

\usepackage{graphicx} % Required for inserting images
\usepackage{mathtools,amsmath,amssymb,amsthm}
\usepackage{comment,xcolor}
\usepackage{multirow}
\usepackage{makecell}
\usepackage{array,booktabs}
\usepackage{color, colortbl}
\usepackage{kbordermatrix}
\usepackage{tabstackengine}[2016-11-30]
\usepackage{xcolor}
\usepackage[labelfont=bf]{caption}
\usepackage{algorithm}
\usepackage{algpseudocode}
\usepackage{siunitx}
\usepackage[affil-it]{authblk}
\newtheorem{theorem}{Theorem}[section]
\newtheorem{proposition}[theorem]{Proposition}

\theoremstyle{definition} % This theoremstyle is non-italic
\newtheorem{definition}[theorem]{Definition}

\newtheorem{example}[theorem]{Example}
\renewenvironment{proof}{{\noindent\bfseries Proof.}}{\qed} %to get bf proof
\usepackage[colorlinks=true,allcolors=blue!75!black,backref=page]{hyperref}
\definecolor{ThistleBlue}{rgb}{102, 153, 255}
\newcommand{\rnk}{\operatorname{rank}}

\title{Exploring the space of graphs with fixed discrete curvatures}

\author[1]{Michelle Roost}
\author[1]{Karel Devriendt\thanks{Contact: karel.devriendt@mis.mpg.de}}
\author[1]{Giulio Zucal}
\author[1,2,3]{J\"{u}rgen Jost}
\affil[1]{Max Planck Institute for Mathematics in the Sciences, Leipzig, Germany}
\affil[2]{Santa Fe Institute, Santa Fe, New Mexico, USA}
\affil[3]{ScaDS.AI, Leipzig, Germany}

\date{\today}

\begin{document}
\maketitle

%%%%
\begin{abstract}
Discrete curvatures are quantities associated to the nodes and edges of a graph that reflect the local geometry around them. These curvatures have a rich mathematical theory and they have recently found success as a tool to analyze networks across a wide range of domains. In this work, we consider the problem of constructing graphs with a prescribed set of discrete edge curvatures, and explore the space of such graphs. We address this problem in two ways: first, we develop an evolutionary algorithm to sample graphs with discrete curvatures close to a given set. We use this algorithm to explore how other network statistics vary when constrained by the discrete curvatures in the network. Second, we solve the exact reconstruction problem for the specific case of Forman--Ricci curvature. By leveraging the theory of Markov bases, we obtain a finite set of rewiring moves that connects the space of all graphs with a fixed discrete curvature.
\end{abstract} 
%%%%

%%%%
\section{Introduction}\label{section: introduction}

%{\color{blue}  Trugenberger, Universe 9, 499 (2023)   \cite{universe9120499}, Trugenberger, JHEP 9, 045 (2017)  \cite{trugenberger2017combinatorial}, Reconstructing degree distribution and triangle counts from edge-sampled graphs  \cite{arnold2022reconstructing}}

Networks are used ubiquitously as a tool to model, analyze and design complex systems in a wide range of domains of science and technology such as neuroscience \cite{fornito_2016_fundamentals}, epidemiology \cite{pastor-satorras_2015_epidemic}, economics \cite{hausmann_2013_atlas, battiston_2012_debtrank}, urban systems \cite{barthelemy_2016_structure}, energy infrastructure \cite{pagani_2013_power} and many more. Across these different disciplines, networks, graphs and their higher-order generalizations \cite{lambiotte_2019_networks, battiston_2021_physics, bick_2023_higher} provide a unifying mathematical language that allows to translate concepts and tools from one setting to another. Moreover, a series of observations at the turn of the 21st century found that the networks extracted from real-world systems often exhibit common features: heavy-tailed degree distributions \cite{barabasi1999emergence}, a ``small world" structure \cite{watts1998collective} and rich mesoscale connectivity patterns such as communities \cite{girvan2002community}. This initiated an ongoing pursuit to find models and theories that explain %why these features are as common as they are, 
the presence of these features, and develop new mathematical methods to detect, describe and quantify them. We point to \cite{broido_2019_scalefree,voitalov_2019_scalefree, cantwell_2020_thresholding} for a critical perspective on these observations.

Among the many approaches to study complex networks, this article is concerned with geometry and more specifically with discrete curvature as a tool to describe the structure of graphs; see \cite{boguna_2021_network, Wu2014EmergentCN} for an overview of different notions of geometry associated with complex systems and their applications. We now explain in more detail the main topic of the article: discrete curvature.
\\~\\
%Curvature is a concept from differential geometry, the study of smooth objects such as surfaces and manifolds, that quantifies how much these objects differ locally from being flat. While classical definitions of curvature need the structure of a Riemannian manifold, many works in recent years have generalized these definitions so they can be applied to other metric spaces such as graphs; this results in notions of discrete curvature --- see Section \ref{section: background, FR and AFR} for concrete examples. 
Curvature is a concept from Riemannian geometry, the study of smooth objects such as surfaces equipped with a metric, and quantifies how much these objects differ locally from being flat. While classical definitions of curvature need the structure of a Riemannian manifold, more recent works have generalized these definitions so they can be applied to other metric spaces such as graphs; this results in discrete notions of curvature --- see Section \ref{section: background, FR and AFR} for a concrete example. 

Notions of discrete curvature that can be applied to graphs were proposed, amongst others, by Forman based on Laplacian operators \cite{forman2003bochner}, by Ollivier, with a modification by Lin, Lu and Yau, based on neighborhood overlap \cite{ollivier2007ricci,ollivier2010survey, OllivTheor2Lin}, by Bakry and \'{E}mery based on the Bochner identity \cite{bakry2013analysis}, by Erbar and Maas based on a certain entropy functional \cite{erbar_2012_ricci}, by Steinerberger based on solutions to an equilibrium problem \cite{steinerberger_2023_curvature}, and by Devriendt and Lambiotte based on the effective resistance \cite{Devriendt_2022,devriendt2023graph}. Importantly, many of these discrete curvatures satisfy properties and results analogous to the powerful theorems from differential geometry; see for instance \cite{OllivTheor1Lin,OllivTheor2Lin, OllivTheor3Bauer,OllivTheor4Jost, OllivTheor5Loisel, forman2003bochner, liu_2018_bakry, devriendt2023graph, steinerberger_2023_curvature}. Of particular relevance to the central question of emergent structure and organization in complex systems, we note that many notions of discrete curvature satisfy powerful `local-to-global' theorems, where local constraints on the curvature lead to statements on the global structure of the graph.

As mentioned, most notions of discrete curvature naturally apply to more general spaces than graphs. This is also the case for the curvature that we mainly focus on in this article, Forman's discrete Ricci curvature, which was defined in \cite{forman2003bochner} in the setting of CW  complexes. Graphs are special cases of such complexes. 

Aside from a rich mathematical theory, discrete curvature has also found successful applications as a tool to analyze and characterize the structure of complex real-world networks. It has been used for instance in the context of community detection \cite{ComDet,Gosztolai2021UnfoldingTM,fesser_2023_augmentations}, in the study of biological networks \cite{OllivAppcancNet, Eidi2020EdgebasedAO}, in characterizing connectivity patterns in brain networks \cite{CurvAutism}, in studying market instabilities in financial networks \cite{CurvMarketInstab}, and to detect and address the problem of oversquashing in graph neural networks \cite{topping_2022_understanding,bober2022rewiring}. Finally, graph curvature also plays an important role in combinatorial quantum gravity, that is, in attempts to unify Einstein’s theory of general relativity with the models of elementary particle physics by using a random graph as a version of a discrete space-time structure at the Planck scale, see \cite{trugenberger2017combinatorial,universe9120499}.
\\
~
\\
%Both the mathematical study of discrete curvature as well as its use in applications have so far been concerned with computing and studying the discrete curvatures of a given graph or family of graphs. 
In this article, we consider the following question: %\textit{How to construct graphs with a given set of discrete curvatures?} 
\textit{How to construct graphs with prescribed discrete curvatures?} To our knowledge, this inverse problem has not yet been addressed for any notion of discrete curvature, and we consider it a basic methodological gap in trying to understand how the geometry of a graph or network influences its other features. As a first step in addressing this methodological gap, this article provides two different graph reconstruction approaches for specific choices of discrete curvatures. We furthermore explore, both experimentally and theoretically, the resulting space of graphs with fixed discrete curvatures.

In the first part, Section \ref{section: approximate sampling}, we consider the problem of \emph{approximate reconstruction} and develop a methodology that works for any notion of discrete curvature. We use a Markov chain Monte Carlo-type algorithm (MCMC) to iteratively construct graphs with discrete curvatures close to a target set of curvatures. We then use this algorithm to explore the ensemble of graphs with a given discrete curvature sequence in a number of experiments. For computational reasons, our discussion focuses on two simple notions of discrete curvature, Forman--Ricci and augmented Forman--Ricci curvature, but we have verified the methodology for other curvatures as well.

In the second part, Section \ref{section: markov bases}, we take a more theoretical approach and consider the problem of \emph{exact reconstruction} of graphs from a given target Forman--Ricci curvature sequence. In this specific case, the problem reduces to a combinatorial question involving the so-called joint degrees in the graph; see for instance \cite{stanton_2012_constructing}. Using known characterizations of possible joint degrees in combination with the theory of Markov bases \cite{MarkovBases}, we show that there exists a finite set of graph rewiring moves that connect any two graphs with the same Forman--Ricci curvatures and node degrees; this is the content of Theorem \ref{th: markov moves and transpositions}. The set of rewiring moves can be computed using existing computer algebra packages, and we give explicit examples for the case of graphs with a small maximum degree. %Our result naturally leads to a Markov chain that samples graphs with Forman--Ricci curvatures and degrees equal to a given graph. 
Our result naturally leads to an ergodic Markov chain over the set of all graphs with equal Forman--Ricci curvatures and degrees. 

After finishing our work, we found that the recent preprint \cite{almendrahernández2024irreducible} %considered a similar problem and approach, using Markov bases to explore the space of {
considers, in a different context, a question related to our \emph{exact reconstruction problem} and develops a similar approach. The authors make use of Markov basis techniques to explore the space of node-colored graphs that have fixed node degrees and a fixed number of edges whose end points have a given pair of colors. They prove results on the graph rewiring moves (i.\,e., the Markov basis) that solves this problem and, in the case of simple graphs, show that the complexity of these moves grows with the number of colors; this is in line with our experimental observation in Theorem \ref{th: markov moves and transpositions} that the number of rewiring moves grows quickly with the maximum degree. %containing vertex pairs of given colors. 
The problem in \cite{almendrahernández2024irreducible} is distinct from ours, and their approach and solution does not address the problem considered in this article, but we believe that the concurrency of both articles illustrates the potential of using Markov basis techniques to explore spaces of graphs with various properties and constraints. %We remark that the problem considered in \cite{almendrahernández2024irreducible} and our problem of constructing a set of moves preserving the degree sequence and the discrete curvature sequence are different. Only the promising idea of using Markov bases to study moves that change a graph while preserving some given constraints is similar.}
\\
~
\\
\textbf{Organization:} The rest of this article is organized as follows. Section \ref{section: background, FR and AFR} introduces the relevant definitions from graph theory and discrete curvature and formalizes the reconstruction question. In Section \ref{section: approximate sampling}, we propose an evolutionary MCMC algorithm for approximate reconstruction from a target curvature sequence. We then use the algorithm in a number of experiments and describe network statistics of graphs sampled using the algorithm. In Section \ref{section: markov bases}, we consider the exact reconstruction problem. We formalize the problem and solve it using the theory of joint degree matrices and Markov bases. Section \ref{section: conclusion} concludes the article with a summary of the results and an outlook towards future applications and extensions.
%%%
%%%
\section{Background and Forman--Ricci curvatures}\label{section: background, FR and AFR}
%%%
We consider simple, unweighted, undirected, %locally 
finite graphs $G=(V,E)$ where $V$ is the set of \textit{nodes} and $E$ is the set of \textit{edges}, which are pairs $\lbrace u,v\rbrace\in E$ of distinct nodes that are connected in the graph. As mentioned in the introduction, we will characterize the structural properties of graphs by considering their discrete curvature. Our focus will be on two definitions of discrete Ricci curvature on edges.
%Version 2: Here we first give formula 7.1 for general edges in a CW complex, and only later consider simplicial complexes
Forman--Ricci curvature was originally defined for CW complexes \cite{forman2003bochner}; roughly speaking, these are topological spaces made up of a collection of `cells' (topological balls) of different dimensions, with cells of dimension $k$ attached along their boundary to cells of dimension $k-1$, in a recursive way. Formula (7.1) in \cite{forman2003bochner} gives the curvature of a $1$-dimensional cell $e$, i.e., an edge, in a CW complex as
\begin{equation}\label{fr1}
  F(e) = \sharp (f^{(2)}>e)+2 -\sharp (\text{parallel neighbors of }e).
\end{equation}
Here, $f^{(2)}>e$ means that $f^{(2)}$ is a 2-cell containing the edge $e$ in its boundary, and an edge $\tilde{e}\neq e$ is a parallel neighbor of $e$ if either $\tilde{e}$ shares a node with $e$, or there is a $2$-cell containing both $e$ and $\tilde{e}$ in its boundary, but not both. In the case of an unweighted, undirected, 
%locally 
finite graph $G$, which is a 1-dimensional CW complex and thus does not contain any 2-cells, all edges that share a node with $e$ are %considered 
parallel neighbors. Consequently, for an edge $e=\{u,v\}$ in a graph, equation \eqref{fr1} becomes
\begin{equation}\label{eq: definition Forman curvature}
F(e) = 4 - \deg(u) - \deg(v).
\end{equation}
The degree $\deg(u)$ counts all edges containing $u$, that is, $e$ itself and the neighbors of $e$ that share the node $u$. Expression \eqref{eq: definition Forman curvature} is called the \textit{Forman--Ricci curvature} of $e$ in the graph.

The clique complex of a graph is the simplicial complex (and thus CW complex) whose simplices are given by all complete subgraphs (cliques) in the graph. In particular, the $2$-cells in the clique complex of a graph are given by the triangles in the graph. Since the edges that form a triangle with $e$ are not parallel edges of $e$, equation \eqref{fr1} for an edge $e=\{u,v\}$ in the clique complex of a graph becomes
\begin{equation}\label{eq: definition augmented Forman curvature}
F^{\#}(e) = 4 - \deg(u) - \deg(v) + 3\mathcal{T}(e),
\end{equation}
where $\mathcal{T}(e)$ is the number of triangles in the graph that contain $e$. Expression \eqref{eq: definition augmented Forman curvature} was called the \emph{augmented Forman--Ricci curvature} of $e$ in the graph $G$ in \cite{Areejit_2018_comparative}.

The quantity $F^{\#}$ is often used because it can be computed efficiently while still capturing the local geometry of a graph well; for instance, it correlates highly with other notions of discrete curvature such as Ollivier--Ricci curvature \cite{Areejit_2018_comparative} and Bakry--\'{E}mery curvature \cite{mondal_2024_bakry}. An alternative approach is to consider the Forman--Ricci curvature of an edge in the $2$-dimensional CW complex obtained by adding a $2$-dimensional cell for every cycle in the graph, up to a certain length. For cycles up to length three this corresponds to $F^{\#}$, and when considering cycles of length four and more, this involves a tradeoff between expressivity and computational complexity 
see for instance \cite{jost_2021_characterizations, tee_2021_enhanced, fesser_2023_augmentations} for more on this. Figure \ref{fig: example FR AFR} shows a small graph with the Forman--Ricci and augmented Forman--Ricci curvatures computed for all edges.
\begin{figure}[h!]
    \centering
    \includegraphics[width=0.6\textwidth]{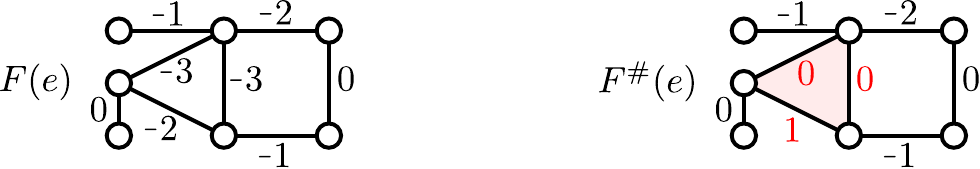}
    \caption{The Forman--Ricci curvatures $F(e)$ and augmented Forman--Ricci curvatures $F^{\#}(e)$ in a graph. The differences due to the presence of a triangle are highlighted in red.}
    \label{fig: example FR AFR}
\end{figure}
\\
For a given graph $G$, the \textit{curvature sequence} is the multiset (set with possible repeats) of all discrete curvatures of its edges. For instance, the Forman--Ricci curvature sequence is the multiset $\lbrace F(e)\rbrace_{e\in E}$, which for the example in Figure \ref{fig: example FR AFR} equals $\lbrace -3,-3,-2,-2,-1,-1,0,0\rbrace$. We can now make the main methodological question of this article more precise:
\begin{equation}\tag{Q1}\label{eq: main question}
\text{\emph{Given a multiset $\mathcal{C}^{\star}$ of numbers, construct a graph with curvature sequence $\mathcal{C}^\star$.}}
\end{equation}
The answer to this question can then be used to address a second goal of the article: to explore the space of all graphs that have a given discrete curvature. 

Section \ref{section: approximate sampling} addresses question \eqref{eq: main question} algorithmically by proposing an algorithm that samples graphs $\hat{G}$ such that $\mathcal{C}(\hat{G})\approx \mathcal{C}^{\star}$, for any notion of discrete curvature. %This approach works for any notion of discrete curvature, but we focus on Forman--Ricci and augmented Forman--Ricci curvature for computational reasons. 
Section \ref{section: markov bases} addresses question \eqref{eq: main question} theoretically for the case of Forman--Ricci curvature. We show that there exists a finite set of graph operations, such that any two graphs $G,G'$ with the same Forman--Ricci curvature sequence $\mathcal{C}(G)=\mathcal{C}(G')$ and degree sequence can be transformed into each other by using these graph operations. Throughout the article, we limit our scope to choices of $\mathcal{C}^\star$ for which there exists at least one graph $G$ with curvature sequence $\mathcal{C}^\star$. Characterizing the multisets that are realizable as discrete curvature sequences of a graph is an interesting problem in itself, and it is still unexplored for most notions of discrete curvature.
%%%%
%%%%
\section{Approximate graph reconstruction from curvatures}\label{section: approximate sampling}
The main methodological question of this article \eqref{eq: main question} asks to construct graphs with predefined discrete curvatures. Similar to reconstructing graphs from statistics such as motif counts or eigenvalues \cite{ipsen_2002_evolutionary}, this is a very difficult task in general and it highly depends on the choice of discrete curvature. Therefore, this section proposes an algorithm for \emph{approximate reconstruction} that works for any choice of curvature. For a given multiset $\mathcal{C}^\star$, the algorithm samples graphs $\hat{G}$ such that $\mathcal{C}(\hat{G})\approx\mathcal{C}^\star$. We use this algorithm to explore the ensemble of graphs with curvatures close to a given target sequence. In particular, we study how a number of standard network summary statistics vary over this ensemble. 
%%%
\subsection{Evolutionary algorithm for sampling graphs with given curvatures}\label{subsection: description of algorithm}
We follow the methodology of \cite{ipsen_2002_evolutionary} and use an evolutionary Markov chain Monte Carlo-type algorithm (MCMC) that consists of three steps: \textit{initialization, mutation and selection}. The algorithm starts with an initial graph $G_0$ which is then evolved iteratively to make its curvature sequence gradually closer to the target sequence. After choosing a notion of discrete curvature of interest, the algorithm takes three inputs: (i) the number of nodes $n$ and the target sequence $\mathcal{C}^\star$, which is a multiset of numbers, (ii) a parameter $\theta$ which introduces variability to avoid the algorithm from getting stuck and (iii) a parameter $T$, which is the number of steps after which the algorithm halts. In practice we choose the parameters $\theta$ and $T$ for a given sequence using grid search. %, before sampling many graphs. 
We now describe the algorithm steps in more detail.
\\
\textbf{$\blacktriangleright$ Initialization:} Construct a random graph $G_0$ with $n$ nodes and $m:=\vert\mathcal{C}^\star\vert$ edges, starting from a random tree (to guarantee connectedness) and with the remaining edges randomly added while avoiding multi-edges. The iteration parameter $t$ is set to zero.
\\
\textbf{$\blacktriangleright$ Mutation:} Select a random edge $\lbrace u,v\rbrace\in E(G_t)$ in the graph and construct a graph $G'$ in which the edge $\lbrace u,v\rbrace$ is rewired to $\lbrace u, v'\rbrace$, where $v'$ is chosen uniformly at random while avoiding isolated nodes and multi-edges.
\\
\textbf{$\blacktriangleright$ Selection:} Compute the discrete curvature sequences $\mathcal{C}(G_t)$ and $\mathcal{C}(G')$ and compare with the target curvature sequence by computing the mean squared error\footnote{We found that choosing a different measure to compare the curvature sequences does not have much influence on the convergence speed of the algorithm.}:
$$
\text{MSE}_t = \frac{1}{m}\sum_{i=1}^m(\mathcal{C}(G_t)_i - \mathcal{C}^\star_i)^2\quad\text{~and~}\quad\text{MSE}' = \frac{1}{m}\sum_{i=1}^m(\mathcal{C}(G')_i - \mathcal{C}^\star_i)^2,
$$
where the curvature sequences are sorted in ascending order. The mutation $G'$ is accepted if it results in a decrease in the mean squared error $\Delta\text{MSE}=\text{MSE}'-\text{MSE}_t<0$, or otherwise it is accepted with probability $p=\exp(- \Delta\text{MSE}/(\text{MSE}_t \cdot \theta))$. Finally, set $t\leftarrow t+1$ and $G_t\leftarrow G'$ if the mutation is accepted or $G_t\leftarrow G_{t-1}$ otherwise, and stop the algorithm after $T$ total iterations. %and continue the process until $t=T$.

Figure \ref{fig:loss_evolution_celegans} illustrates the evolution of the curvature sequences throughout the algorithm. It shows the histogram of the target sequence $\mathcal{C}^\star$ on the top left and an example of the evolution of the loss function with respect to the number of accepted mutations on the bottom %right
left. The label A denotes the initial unmutated random graph $G_0$ as explained above in the %key point 
\emph{Initialization} step. The label D %is 
denotes the final mutated graph %$\mathcal{C}(G_T)$ 
obtained after $T=\num{2e6}$ total iterations, of which $\num{16e4}$ mutations were accepted. The respective curvature histograms from A to D are shown in the blue box on the right side of the figure. With respect to the target curvature sequence, the MSE evolves from the random starting configuration A with $\text{MSE}_{\text{A}} \approx 8469$, over B and C with $\text{MSE}_{\text{B}} \approx 6269$, and $\text{MSE}_{\text{C}} \approx 437$, to the final evolved configuration with $\text{MSE}_{\text{D}} \approx 0.29$. We note the good convergence of the final curvature sequence $\mathcal{C}(G_T)$. 

While this methodology is applicable to any choice of discrete curvature, we focus on Forman--Ricci and augmented Forman--Ricci curvature for computational reasons in our experiments. We found that the algorithm still converges for Ollivier--Ricci curvature and resistance curvature target sequences, but with prohibitively long computation times.
\sisetup{tight-spacing=true}
\begin{figure}[h!]
\centering
\includegraphics[width=1\textwidth]{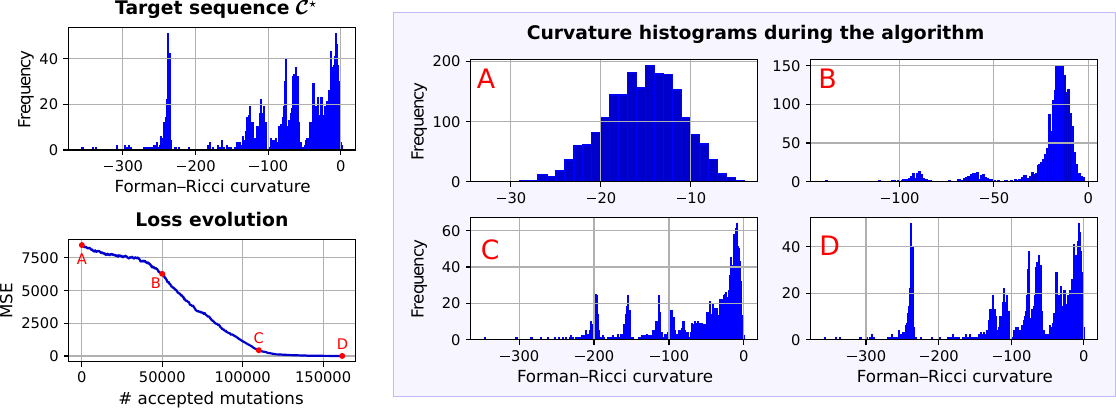}
\caption{{Evolution of curvature sequences throughout the sampling algorithm.} \textbf{Top left:} Histogram of the target sequence $\mathcal{C}^\star$. This is the Forman--Ricci curvature sequence of the \emph{C. elegans} metabolic network (see Section \ref{subsubsection: target sequences}). \textbf{Bottom left:} Evolution of the mean squared error MSE$_t$ between the target sequence and the Forman--Ricci curvature sequences. \textbf{Highlight box:} Histogram of Forman--Ricci curvatures $\mathcal{C}(G_t)$ at four points during the algorithm: (A) the initial graph (B) after $\num{5e4}$ accepted mutations, (C) after $\num{11e4}$ accepted mutations and (D) at the end of the algorithm, after $\num{16e4}$ accepted mutations and $T=\num{2e6}$ total iterations.}
\label{fig:loss_evolution_celegans}
\end{figure}
%%%
\subsection{Approximate sampling experiments}\label{subsection: sampling experiments}
%%%
We now use the algorithm described in Section \ref{subsection: description of algorithm} to explore ensembles of graphs with a given discrete curvature sequence. For both Forman--Ricci and augmented Forman--Ricci curvature, we consider four target curvature sequences and sample $1000$ graphs. We then compute a number of network summary statistics and explore how they vary over these graphs.
%%%
\subsubsection{Target sequences}\label{subsubsection: target sequences}
%%%
To guarantee target sequences which are realizable, we use curvature sequences of existing graphs. We consider graphs obtained from three widely-studied random graph models \cite{newman2010networks} and one real-world network. The considered graphs are as follows:
\begin{itemize}
    \item \textbf{Erd\H{o}s--R\'{e}nyi (ER) random graph.}
     The graph has $n=500$ nodes and every pair of nodes is connected independently with probability $0.01$. These graphs result in a \emph{unimodal target sequence}.
    % \newpage \noindent
    \item \textbf{Stochastic block model (SBM) random graph}. The graph has $n=500$ nodes divided in two equal sized groups. Pairs of nodes in group $1$ are connected with probability $p=0.05$, in group $2$ with probability $p=0.01$ and pairs of nodes from different groups with probability $p=0.003$. These graphs result in a \emph{bimodal target sequence}.
    
    \item \textbf{Barab\'{a}si--Albert (BA) random graph.} The graph has $n=500$ nodes and is constructed by starting from a star graph with four nodes and then adding new nodes. Each new node is connected to three of the existing nodes with probability proportional to their degrees. These graphs result in a \emph{heavy-tailed target sequence}.    

    \item \textbf{Real-world network}. We consider the metabolic network of the \emph{Caenorhabditis elegans} worm (\emph{C. elegans}) \cite{kunegis2013konect,duch2005community,jeong2000large,overbeek2000wit}, this is a well-studied graph in biology. This graph has $n=453$ nodes and $m=2025$ edges and results in a \emph{heterogeneous/complex target sequence}.
\end{itemize}

\begin{figure}[h!]
\centering
\includegraphics[width=0.95\textwidth]{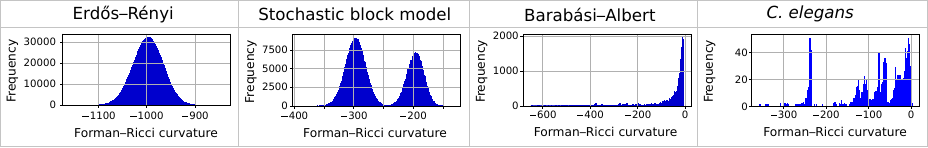}
\caption{Histograms of the Forman--Ricci curvature of an Erd\H{o}s--R\'{e}nyi, stochastic block model and Barab\'{a}si--Albert random graph and the \emph{C. elegans} metabolic network.}
\label{fig:example_forman_random_models}
\end{figure}
%%%
\subsubsection{Summary statistics}
%%%
To summarize the structure of the graphs obtained from our algorithm, we compute four global network statistics: 
\begin{itemize}
    \item \textbf{Number of triangles:} the number of triangles in a graph is a simple but important summary statistic. It contains information about the density of the graph and it is closely related to transitivity \cite{wasserman1994social} and discrete curvature \cite{OllivTheor4Jost}.

    \item \textbf{Average clustering coefficient:} the clustering coefficient $C_v$ of a node $v$ is the number of triangles that contain $v$ divided by the number of possible triangles $\binom{\deg(v)}{2}$. The average clustering coefficient $\frac{1}{n}\sum C_v$ is an important graph property that reflects its geometry \cite{OllivTheor4Jost, krioukov_2016_clustering}.

    \item \textbf{Average shortest path length:} this is the average, over all pairs of nodes $u$ and $v$, of the smallest number of steps between $u$ and $v$. It is an intuitive summary of the metric structure of a graph.

    \item \textbf{First positive eigenvalue:} the normalized graph Laplacian is the $n\times n$ matrix \linebreak $L=I-D^{-1}A$ where $I$ is the identity matrix, $D$ is the diagonal degree matrix and $A$ is the adjacency matrix. The eigenvalues of the Laplacian contain a lot of information about the graph structure and dynamical processes such as diffusion and random walks associated with the graph \cite{chung_spectral_1997, mulas_2022_graphs}. In particular, $L$ is a positive semidefinite matrix and its first positive eigenvalue $\lambda_2$ quantifies how well-connected the graph is.
    \end{itemize}
%%%
\subsubsection{Experiments}
%%%
Before sampling the ensemble of graphs, we determine the optimal parameters $\theta$ and $T$ for each of the target sequences. Table \ref{table:theta_gridsearch} shows the parameter values obtained via grid search. %When $\theta$ is chosen much higher or lower, the algorithm does not converge well.  
\begin{table}[h!]
\renewcommand{\arraystretch}{1.5}
\centering
% \begin{tabular}{|>{\centering\arraybackslash}p{3.7cm}||>{\centering\arraybackslash}p{3.5cm}|>{\centering\arraybackslash}p{4cm}|}
% \begin{tabular}{|c|c|c|c|c|}
% \hline
%  & ER & SBM & BA & \emph{C. elegans}
%  \\
%  \hline
%  \makecell{Forman--Ricci curvature\\ sequence}
%  %& 0.003 & 0.003 & 0.003 & 0.0008
%  & \num{3e-3} & \num{3e-3} & \num{3e-3} & \num{8e-4}
%  \\
%  \hline
%  \makecell{augmented Forman--Ricci\\ curvature sequence}
%  %& 0.003 & 0.004 & 0.003 & 0.0006 %different because BA-SBM swapped
%  & \num{3e-3} & \num{3e-3} & \num{4e-3} & \num{6e-4}
%  \\
%  \hline
% \end{tabular}
%%%%%%%%%%%%%%%%%%%%%%%%%%%%%%%%%%%%%%%%%%%%%%%%%%%%%%%%%%%%%%%
% \begin{tabular}{|c|c|c|c|c|c|c|c|c|}
% \hline
%  & \multicolumn{2}{c|}{ER} & \multicolumn{2}{c|}{SBM} & \multicolumn{2}{c|}{BA} & \multicolumn{2}{c|}{\emph{C. elegans}}
%  \\
%  & $\theta$ & $T$ & $\theta$ & $T$ & $\theta$ & $T$ & $\theta$ & $T$ \\
%  \hline
%  \makecell{Forman--Ricci curvature\\ sequence}
%  %& 0.003 & 0.003 & 0.003 & 0.0008
%  & \num{3e-3} & \num{1e6} & \num{3e-3} & \num{1e6} & \num{3e-3} & \num{1e6} & \num{8e-4} & \num{2e6}
%  \\
%  \hline
%  \makecell{augmented Forman--Ricci\\ curvature sequence}
%  %& 0.003 & 0.004 & 0.003 & 0.0006 %different because BA-SBM swapped
%  & \num{3e-3} & \num{2e6} & \num{3e-3} & \num{2e6} & \num{4e-3} & \num{2e6} & \num{6e-4} & \num{2.5e6}
%  \\
%  \hline
% \end{tabular}
%%%%%%%%%%%%%%%%%%%%%%%%%%%%%%%%%%%%%%%%%%%%%%%%%%%%%%%%%%%%%%%
\begin{tabular}{|c|c|c|c|c|c|}
\hline
 \multicolumn{2}{|c|}{\phantom{-}} & ER & SBM & BA & \emph{C. elegans}
 \\
 \hline
 \multirow{2}{*}{\shortstack{Forman--Ricci curvature\\ sequence}}
 %& 0.003 & 0.003 & 0.003 & 0.0008
 & $\theta$ & \phantom{-.}\num{3e-3} & \phantom{-.}\num{3e-3} & \phantom{-.}\num{3e-3} & \phantom{-.}\num{8e-4}
 \\
 & $T$ & \num{1e6} & \num{1e6} & \num{1e6} & \num{2e6}
 \\
 \hline
 \multirow{2}{*}{\shortstack{augmented Forman--Ricci\\ curvature sequence}}
 %& 0.003 & 0.004 & 0.003 & 0.0006 %different because BA-SBM swapped
 & $\theta$ & \phantom{-.}\num{3e-3} & \phantom{-.}\num{3e-3} & \phantom{-.}\num{4e-3} & \phantom{-.}\num{6e-4}
 \\
 & $T$ & \num{2e6} & \num{2e6} & \num{2e6} & \num{2.5e6}\phantom{...}
 \\
 \hline
\end{tabular}
\caption{Optimal $\theta$ and $T$ parameters determined via grid search.}
\label{table:theta_gridsearch}
\end{table}

We now consider the four graphs (ER, SBM, BA, \emph{C. elegans}), compute their Forman--Ricci curvature sequence as a target sequence $\mathcal{C}^\star$ and then compute $1000$ graphs with $\mathcal{C}(\hat{G})\approx\mathcal{C}^\star$ using our algorithm. As an indication of the quality of the approximation, Table \ref{table:mse} shows the mean squared errors for the sampled graphs. Generally, we observe that the MSE increases as the graphs become less uniform, as is the case for the Barab\'{a}si--Albert and the \emph{C. elegans} network.
\\~\\
Figure \ref{fig:graph_properties_forman} shows the histograms of the network statistics (in blue) and the value of this statistic for the starting graph (red line). For each statistic, we also plot the \textit{standard score} $\tfrac{\vert s^\star-\langle s\rangle\vert}{\sigma}$, where $s^\star$ is the value of the statistic for the target graph, $\langle s\rangle$ is the ensemble average of the statistic and $\sigma$ is the standard deviation of the statistic over the ensemble. This number is a proxy for how typical or atypical the statistics of the target graph are with respect to the ensemble.
% \begin{figure}[h!]
% \centering
% \includegraphics[width=1\textwidth]{figures/graph_properties_forman_swapped.pdf}
% %\caption{Histograms of graph properties of the graph ensembles in blue and the value of the corresponding target sequence as red line. The red value is the standard score stating how many standard deviations the target value is away from the mean of the converged graphs. The target sequences are Forman--Ricci sequences of three random models Erd\H{o}s--R\'{e}nyi model (ER), Barab\'{a}si--Albert model (BA), and stochastic block model (SBM), and the metabolic network of \emph{C. elegans} as an example of an empirical network.}
% \caption{Network statistics for graphs sampled using the algorithm from Section [ref], based on four target \textit{Forman--Ricci curvature sequences} (ER, SBM, BA, C. elegans). The histograms of the observed statistics in the sampled graph ensemble are shown in blue and the red line indicates the value of the statistic in the target graph, with the $Z$ value shown in red.} %The target sequences are Forman--Ricci sequences of three random models Erd\H{o}s--R\'{e}nyi model (ER), Barab\'{a}si--Albert model (BA), and stochastic block model (SBM), and the metabolic network of \emph{C. elegans} as an example of an empirical network.}
% \label{fig:graph_properties_forman}       
% % \vspace{-0.4cm}
% \end{figure}

\begin{figure}[h!]
\centering
\includegraphics[width=1\textwidth]{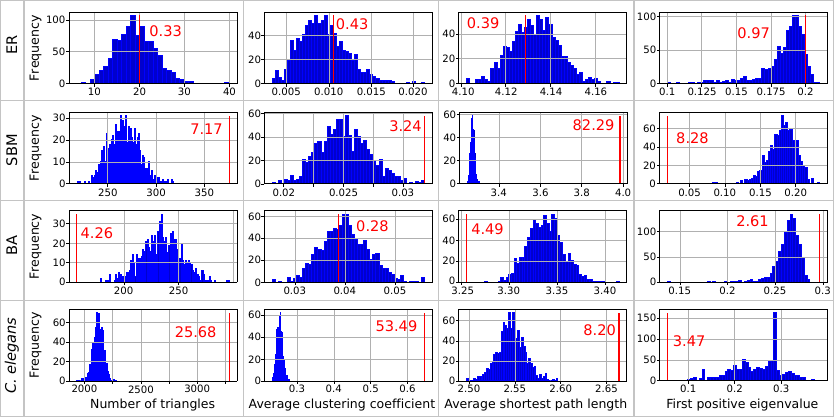}
%\caption{Histograms of graph properties of the graph ensembles in blue and the value of the corresponding target sequence as red line. The red value is the standard score stating how many standard deviations the target value is away from the mean of the converged graphs. The target sequences are Forman--Ricci sequences of three random models Erd\H{o}s--R\'{e}nyi model (ER), Barab\'{a}si--Albert model (BA), and stochastic block model (SBM), and the metabolic network of \emph{C. elegans} as an example of an empirical network.}
\caption{Network statistics for graphs sampled using the algorithm from Section \ref{subsection: description of algorithm}, based on four target \textit{Forman--Ricci curvature sequences} (ER, SBM, BA, \emph{C. elegans}). The histograms of the observed statistics in the sampled graph ensemble are shown in blue and the red line indicates the statistic of the target graph together with its standard score.} %The target sequences are Forman--Ricci sequences of three random models Erd\H{o}s--R\'{e}nyi model (ER), Barab\'{a}si--Albert model (BA), and stochastic block model (SBM), and the metabolic network of \emph{C. elegans} as an example of an empirical network.}
\label{fig:graph_properties_forman}       
% \vspace{-0.4cm}
\end{figure}

We note a few observations: the standard scores for the ER-based target sequence are all below $1$. This suggests that the ensemble of graphs obtained from the algorithm is not very different from an ensemble of ER random graphs. Second, for the SBM-based target sequence, we see that the average shortest path length is much smaller for the sampled graphs than for the target graph, and that the first positive eigenvalue is much larger. Both observations suggest that the block structure of the SBM graph, i.\,e., two groups of nodes which are poorly connected, is not well-reproduced by the sampled graphs. Finally, for the \emph{C. elegans}-based target sequence, we see that all standard scores are very high. This implies that the structure of this graph is very different from the sampled graphs. This is not surprising since only fixing the degree sums of edges is a very local and coarse structural constraint and therefore, discrete curvature alone is not capable of reproducing the other, more global graph statistics. Our current approach could be extended by including further graph statistics and properties in the loss function, which should result in sampling from an ensemble that more closely resembles the target graph.
%However, the algorithm is general and can be extended using further graph statistics.}
\\~\\
We repeat the same experiment with the augmented Forman--Ricci curvature. The mean squared error between the sampled and target curvature sequences are shown in Table \ref{table:mse} and the MSE is again higher for the Barab\'{a}si--Albert graph and the \emph{C. elegans} network. The results of the experiment are shown in Figure \ref{fig:graph_properties_augmented}.
\begin{figure}[h!]
\centering
\includegraphics[width=1\textwidth]{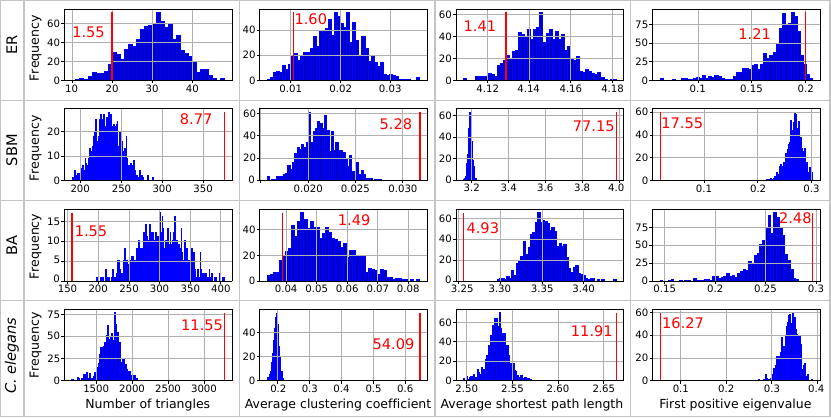}
\caption{Network statistics for graphs sampled using the algorithm from Section \ref{subsection: description of algorithm}, based on four target \textit{augmented Forman--Ricci curvature sequences} (ER, SBM, BA, \emph{C. elegans}). The histograms of the observed statistics in the sampled graph ensemble are shown in blue and the red line indicates the statistic of the target graph together with its standard score.}
\label{fig:graph_properties_augmented}       
% \vspace{-0.4cm}
\end{figure}

We note a few observations: compared to Figure \ref{fig:graph_properties_forman}, the standard scores for the ER-based target sequence have all increased. This is not unexpected, since the graphs are sampled based on the augmented Forman--Ricci curvature which takes into account the presence of triangles in the graph, whereas ER graphs have no correlations between different edges. Secondly, the same observation as before holds for the SBM-based target sequence: the smaller shortest path lengths and larger first positive eigenvalue suggest that the two-block structure of the SBM is not reproduced in the sampled graphs. Similar to the first experiments, using the Forman--Ricci curvature for graphs and including the triangle count is a local and unrefined structural constraint and is therefore not capable of reproducing all graph statistics.

\begin{table}[h!]
\renewcommand{\arraystretch}{1.5}
\centering
\begin{tabular}{|c|c|c|c|c|c|c|c|c|}
\hline
\multirow{2}{*}{MSE} & \multicolumn{4}{c|}{Forman} & \multicolumn{4}{c|}{Augmented Forman} \\
 & ER & SBM & BA & \emph{C. elegans} & ER & SBM & BA & \emph{C. elegans} \\
\hline
\hline
min  & 0.0   & 0.016 & 0.038 & 0.111 & 0.018 & 0.153 & 0.195 & 15.539 \\
max  & 0.005 & 0.039 & 0.193 & 0.416 & 0.191 & 0.285 & 3.757 & 21.469 \\
mean & 0.001 & 0.027 & 0.088 & 0.033 & 0.094 & 0.212 & 1.115 & 18.188 \\
\hline
\end{tabular}
\caption{Minimal, maximal and mean of the MSE between $\mathcal{C}^\star$ and $\mathcal{C}(G_T)$ for the $1000$ sampled graphs in the experiments of Figures \ref{fig:graph_properties_forman} and \ref{fig:graph_properties_augmented}.}
\label{table:mse}
\end{table}

Finally, we note that one should take care in interpreting the histograms of network statistics, since the specific design of our algorithm does not give much control or insight into the distribution from which we sample graphs. Furthermore, we do not have any mixing time results on which we can base our choice of stopping parameter $T$. This makes it hard to control the influence of the initialization $G_0$ on the ensemble of graphs which we sample from. In future work, these issues could be addressed by designing a more tailored algorithm or by trying to find exact expressions for the probability distribution over graphs, for instance using a soft maximum entropy model \cite{squartini_2017_maximumentropy}.

Our algorithm is designed for reconstructing graphs based on a complete target curvature sequence. In many practical situations, however, such complete knowledge of the curvature, or other properties %any other property, 
of all edges is not realistic or feasible, and further steps have to be taken to account for this incomplete information. For instance, \cite{arnold2022reconstructing} takes a Bayesian approach to address the problem of reconstructing graph statistics such as degree and triangle distributions from incomplete (sampled) edge information alone.
%%%%

%%%%
\section{Exact graph reconstruction from curvatures and degrees}\label{section: markov bases}
In this section, we take a second approach to the main question \eqref{eq: main question} and consider the \emph{exact reconstruction problem} of finding graphs whose discrete curvatures precisely match a given multiset of numbers. This problem clearly depends on the notion of discrete curvature, and we will focus on Forman--Ricci curvature $F(\lbrace u,v\rbrace)=4-\deg(u)-\deg(v)$. For reasons which will be explained later, we not only consider fixed curvatures but also fixed degrees. 

In the main result of this section, Theorem \ref{th: markov moves and transpositions}, we find that there exists a finite set of graph rewiring moves which can transform any two multigraphs with the same Forman--Ricci curvatures and node degrees into each other. This result follows from encoding the degree and curvature constraints algebraically using a joint degree matrix and using the theory of Markov bases. Our approach then decomposes into two parts: first, finding all joint degree matrices that are compatible with the curvature and degree constraints and second, finding all graphs with the same joint degree matrix. 
%%%%
\subsection{Definitions and notation}
%%%%
\noindent The rest of this section will use the following conventions and definitions. A \emph{multigraph} can have multi-edges and self-loops while a \textit{simple graph} cannot, and we write ``graph'' if the distinction is not relevant. For a graph $G=(V, E)$ with maximum degree $\Delta$, we denote the set of nodes with degree $a$ as $V_a = \{v \in V(G): \deg(v) = a\}$. The cardinality $\vert V_a\vert$ of this set is the number of nodes with degree $a$, and we define the \emph{degree frequencies} of a graph as the $\Delta$-length tuple
$$
(\vert V_1\vert\,,\,\vert V_2\vert\,,\, \dots \,,\, \vert V_\Delta\vert) %= (\vert V_a \vert)_{a=1}^{\Delta}
\,=\, (\text{\# nodes of degree $a$})_{a=1}^{\Delta}
$$ 
that contains these cardinalities. %The elements of this tuple represent the number of nodes with degree $a$. 
Example \ref{example: degree and curvature frequencies} below shows this for a small graph.

We denote the set of edges whose end point degrees sum to $\kappa$ as $E_\kappa=\{ \{u,v\}\in E(G):\deg(u)+\deg(v)=\kappa\}$; this is also the set of edges with Forman--Ricci curvature equal to $4-\kappa$ as $E_\kappa = \{\{u,v\}\in E(G):F(\{u,v\})=4-\kappa\}$. The cardinality $\vert E_\kappa\vert$ of this set is thus the number of edges with Forman--Ricci curvature $4-\kappa$. We define the \emph{curvature frequencies} of a graph as the $2\Delta$-length tuple 
$$
(\vert E_1\vert\,,\, \vert E_2\vert\,,\,\dots\,,\,\vert E_{2\Delta}\vert) %(\vert E_\kappa\vert)_{\kappa=2}^{2\Delta}
\,=\, \text{(\# edges with Forman--Ricci curvature $4-\kappa$)$_{\kappa=1}^{2\Delta}$.}
$$
Since $\vert E_1\vert=0$ always, the first entry in the curvature frequencies tuple is zero for every graph. We note that the curvature frequencies of a graph contain the same information as the Forman-Ricci curvature sequence $\mathcal{C}(G)$, but encoded in a different way. Example \ref{example: degree and curvature frequencies} below shows the curvature frequencies computed for a small graph.

\begin{example}\label{example: degree and curvature frequencies}
The graph in Figure \ref{figure: degree and curvature frequencies} below has maximum degree $\Delta=4$ and its degree and curvature frequencies are shown on the right. The underlined third entry $(a=3)$ of the degree frequencies is $2$, which indicates that there are $\vert V_3\vert=2$ nodes with degree $3$. The underlined sixth entry $(\kappa=6)$ of the curvature frequencies is $2$, which indicates that there are $\vert E_6 \vert = 2$ edges with degree sum $\kappa=6$ and thus Forman--Ricci curvature $4-\kappa=-2$. This graph will be the running example in the remainder of this section.
\begin{figure}[h!]
    \centering    \includegraphics[width=0.8\textwidth]{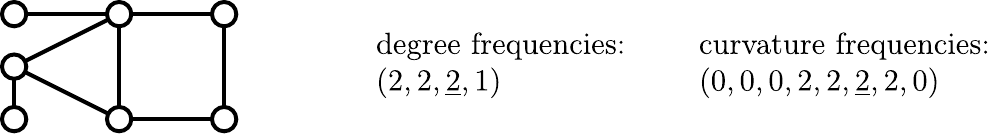}
    \caption{A graph with its degree and curvature frequencies.}
    \label{figure: degree and curvature frequencies}
\end{figure}
\end{example}
%%%
%%%
\subsection{Joint degree matrix and transpositions}
The Forman--Ricci curvature of an edge depends on the sum of the degrees of its end points. This information can be represented by the following matrix: the \emph{joint degree matrix} (JDM) of a graph $G$ is a symmetric $\Delta\times\Delta$ matrix $J$ with entries
\begin{align*}
J_{ab}=J_{ba} &= \left\vert\big\lbrace \lbrace u,v\rbrace\in E(G)\,:\, \lbrace\deg(u),\deg(v)\rbrace=\lbrace a,b\rbrace\big\rbrace\right\vert
\\
&= \text{~\# edges whose end points have degrees $a$ and $b$.}
\end{align*}
We note that the degree and curvature frequencies of a graph can be computed from its JDM: the number of nodes with degree $a$ is equal to the row sum $\vert V_a\vert = \tfrac{1}{a}(\sum_{b} J_{ab}+J_{aa})$ 
and the number of edges with Forman--Ricci curvature $4-\kappa$ is equal to the anti-diagonal sum $\vert E_\kappa\vert=\sum_{a\leq b, a+b=\kappa} J_{ab}$. Figure \ref{fig: example JDM and sequences} shows the JDM and associated frequencies for the graph in Example \ref{example: degree and curvature frequencies}. %a small graph that will be our running example.
\begin{figure}[h!]
    \centering
    \includegraphics[width=1\textwidth]{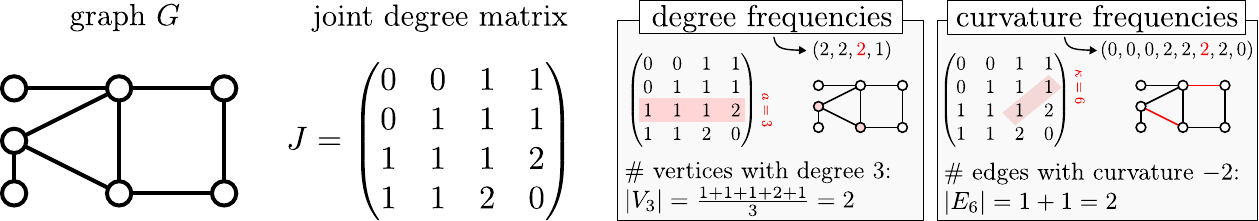}
    \caption{A graph $G$ and its joint degree matrix (JDM). The figures on the right show the degree and curvature frequencies of $G$ and highlight how these can be obtained from the JDM, as follows: $\vert V_3\vert = \tfrac{1}{3}(\sum_{b}J_{3b}+J_{33})$ and $\vert E_6\vert = \sum_{a\leq b, a+b=6} J_{ab}$.}
    \label{fig: example JDM and sequences}
\end{figure}

The following result from \cite{stanton_2012_constructing} gives a complete characterization of JDMs.
\begin{theorem}[{\cite[Thm. 4.2]{stanton_2012_constructing}}]\label{thm: characterization of JDMs}
A symmetric matrix $J\in\mathbb{N}^{\Delta\times\Delta}$ is the JDM of a simple graph if and only if it satisfies, for all distinct $a,b\in\lbrace 1,\dots,\Delta\rbrace$, the conditions
\begin{itemize}
    \item[(i)] $\vert V_a\vert = (\sum_{c=1}^\Delta J_{ac}+J_{aa})/a$\,\, is an integer,
    \item[(ii)] $J_{ab}\leq\vert V_a\vert\cdot\vert V_b\vert$, and
    \item[(iii)] $J_{aa}\leq \binom{\vert V_a\vert}{2}$.
\end{itemize}
\end{theorem}
The authors of \cite{stanton_2012_constructing} gave an algorithm to construct a simple graph from a given JDM. More importantly, they also showed that the set of all graphs with the same JDM is connected through the graph operation shown in Figure \ref{fig: example transposition}: let $\mathcal{E}=\big\lbrace\lbrace u,v\rbrace,\lbrace x,y\rbrace\big\rbrace\subseteq E(G)$ be two edges such that $\deg(u)=\deg(x)$ and let $\mathcal{E}'=\big\lbrace\lbrace u,y\rbrace,\lbrace x,v\rbrace\big\rbrace$. %are not edges in $G$.
The \textit{transposition} of the edges $\mathcal{E}$ in $G$ is a new graph $G'$ obtained by rewiring these two edges. Its nodes and edges are
$$
V(G')=V(G) \quad\text{~and~}\quad E(G')= \big(E(G)\backslash\mathcal{E}\big) \cup \mathcal{E}'.
$$
\vspace{-0.5cm}
\begin{figure}[h!]
    \centering
    \includegraphics[width=0.7\textwidth]{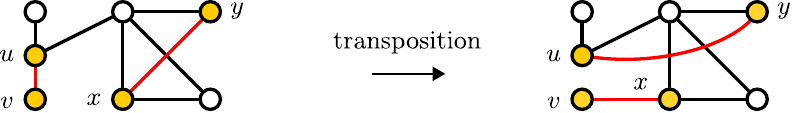}
    \caption{A transposition rewires two edges in a graph without changing the joint degrees. In this example, edges $\lbrace u,v\rbrace,\lbrace x,y\rbrace$ are rewired to $\lbrace u,y\rbrace,\lbrace x,v\rbrace$.}
    \label{fig: example transposition}
\end{figure} 

Transpositions do not change the JDM %from the theory of matchings in convex bipartite graphs, it follows that the set of multigraphs with the same JDM is connected through transpositions \cite{diaconis_2001_statistical,stanton_2012_constructing}. 
of a graph and the theory of matchings in convex bipartite graphs guarantees that the set of multigraphs with the same JDM is connected through transpositions \cite{diaconis_2001_statistical,stanton_2012_constructing}. %In the case of simple graphs, 
In general, transpositions in simple graphs may create multi-edges or self-loops depending on the choice of $\mathcal{E}$, but the following result shows that there always exists a choice for which this is not the case.
\begin{theorem}[{\cite[Thm. 5.2]{stanton_2012_constructing}}]\label{th: transpositions connect simple graphs}
The set of all simple graphs with a given JDM is connected through transpositions.
\end{theorem}
In other words, all simple graphs with the same JDM can be found by performing transpositions, in such a way that every intermediate step is also a simple graph. Since transpositions do not change the JDM, these graphs will also have the same curvature and degree frequencies. This is the first step towards our main result.
\\~\\
We note the following count for the number of multigraphs with a given JDM:
\begin{proposition}
The number of node-labeled multigraphs with a given JDM is equal to
\begin{equation}\label{eq: number of multigraphs} 
\prod_{1\leq a\leq \Delta} \frac{(a\vert V_a\vert)!}{(a!)^{\vert V_a\vert}\cdot 2^{J_{aa}}}\cdot\prod_{1\leq a\leq b\leq \Delta}\frac{1}{J_{ab}!}.
\end{equation}
\end{proposition}
\begin{proof}
We have $\vert V_a\vert$ nodes of degree $a$. Each such node has $a$ ``stubs" that need to be matched to $a$ of the $a\vert V_a\vert$ ``half-edges". For the first node, there are $\binom{a\vert V_a\vert}{a}$ ways to match its stubs, for the second node there are $\binom{a(\vert V_a\vert-1)}{a}$ ways, and so on. This gives a total of
\begin{align*}
&\binom{a\vert V_a\vert }{a}\binom{a(\vert V_a\vert-1)}{a}\dots\binom{a}{a} = \frac{(a\vert V_a\vert)!}{a!(a(\vert V_a\vert-1))!}\frac{(a(\vert V_a\vert-1)!}{a!(a(\vert V_a\vert-2))!}\dots \frac{a!}{a!} = \frac{(a\vert V_a\vert)!}{(a!)^{\vert V_a\vert}}
\end{align*}
ways to connect nodes of degree $a$ to half-edges. This counting procedure distinguishes the half-edges, so to obtain edge-unlabeled multigraphs we have to forget the half-edge labels. This makes any permutation of the $J_{ab}$ edges which are connected to nodes with degree $a$ and $b$ indistinguishable, so we divide the count by $J_{ab}!$. Second, forgetting the half-edge labels makes any flip of the $J_{aa}$ edges connected to two nodes of equal degree $a$ indistinguishable, so we divide the count by $2^{J_{aa}}$. This results in \eqref{eq: number of multigraphs} and completes the proof.
\end{proof}

For the JDM in Figure \ref{fig: example JDM and sequences} with degree frequencies $(2,2,2,1)$, the number \eqref{eq: number of multigraphs} equals $\frac{2!4!6!4!}{2^5(3!)^24!}=30$, so there are $30$ node-labeled multigraphs with this JDM. In practice, we may be interested in simple and non-isomorphic graphs for which this number serves as an upper bound.
%%%
%%%
\subsection{Markov bases}
In the previous section, we found that all graphs with a given JDM are connected by transpositions. %However, not all graphs with the same curvature and degree sequence have the same JDM. 
As a second step, we now want to find all JDMs which are compatible with given curvature and degree frequencies. More precisely, these are all nonnegative integer matrices $J$ that satisfy conditions (i)---(iii) in Theorem \ref{thm: characterization of JDMs} (valid JDMs) and have prescribed row sums (degrees) and anti-diagonal sums (curvatures). We will show how this set of matrices can be obtained using Markov bases.

The theory of Markov bases was developed by Diaconis and Sturmfels in \cite{MarkovBases} to solve the problem of sampling contingency tables, i.\,e., nonnegative integer matrices with fixed row and column sums, in the context of statistics. The main result in this theory is that this sampling problem can be solved using concepts and tools from commutative algebra. We first give the main definition and theorem and then apply it to our problem.
\begin{definition}
Let $A\in\mathbb{Z}^{d\times n}$ be an integer matrix. A \emph{Markov basis for $A$} is a finite set $\mathcal{M}=\lbrace \beta_1,\dots,\beta_M\rbrace\subseteq\ker_{\mathbb{Z}}(A)$ of integer vectors in the kernel of $A$, such that for every $u,v\in\mathbb{N}^n$ with $Au=Av$, there exists a sequence $(i_1,s_1),(i_2,s_2),\dots, (i_L,s_L)\in\lbrace 1,\dots,M\rbrace\times\lbrace\pm 1\rbrace$ that satisfies
\begin{itemize}
    \item[(i)] $u+\sum_{j=1}^L s_{j}\beta_{i_j} = v$, and
    \item[(ii)] $u+\sum_{j=1}^\ell s_{j}\beta_{i_j} \geq 0\text{~for all~}1<\ell<L$.
\end{itemize}
\end{definition}
For $u\in\mathbb{N}^n$ we call $\mathcal{F}_A(u)=\lbrace v\in\mathbb{N}^n : Au=Av\rbrace$ the \emph{fiber of $u$}. The adjective `Markov' reflects that a Markov basis can be used to design a Markov chain on a fiber: points in the fiber are states of the Markov chain and the transitions out of any state $x$ are those vectors $\beta$ in the Markov basis for which $x+\beta\geq 0$. The Markov basis properties guarantee that this chain is ergodic.

The question of how to find a Markov basis for a given integer matrix is addressed by the fundamental theorem of Markov bases (see also Appendix \ref{appendix: background for markov bases}).
\begin{theorem}[Fundamental theorem of Markov bases \cite{MarkovBases}]\label{thm: fundamental theorem of Markov bases}
Let $A\in\mathbb{Z}^{d\times n}$, then $\mathcal{M}\subseteq\ker_{\mathbb{Z}}(A)$ is a Markov basis for $A$ if and only if the binomials $\lbrace \mathbf{x}^{\beta^+}-\mathbf{x}^{\beta^-}:\beta\in\mathcal{M}\rbrace$ generate the ideal $\langle \mathbf{x}^{\beta^+}-\mathbf{x}^{\beta^-}:\beta\in\ker_{\mathbb{Z}}(A)\rangle$.
\end{theorem}
Appendix \ref{appendix: background for markov bases} provides further background on the relevant algebraic concepts, but the details of the theorem are not relevant for the purpose of this article. The most important fact is that a finite Markov basis always exists and that it can be computed using computer algebra systems for sufficiently small matrices $A$. For larger matrices, these computations
can become infeasible in practice and other approaches such as relaxed notions of Markov bases or restrictions to particular fibers may be necessary \cite{almendra-hernandez_2024_markov}.
\begin{example}
Consider the $2\times 4$ integer matrix
$$
A=\begin{pmatrix}
1&1&1&1\\0&1&2&3
\end{pmatrix}.
$$
Using the software package \texttt{4ti2} \cite{4ti2}, we compute that this matrix has the Markov basis:
$$
\mathcal{M}=\left\lbrace\begin{psmallmatrix*}[r]0\\1\\-2\\1\end{psmallmatrix*},
\begin{psmallmatrix*}[r]1\\-2\\1\\0\end{psmallmatrix*},
\begin{psmallmatrix*}[r]1\\-1\\-1\\1\end{psmallmatrix*}\right\rbrace.
$$
By definition of Markov bases, this means that if we have two vectors $u,v\in\mathbb{N}^4$ such that $Au=Av$, then we can obtain $v$ from $u$ by adding and subtracting some sequence of the three vectors in $\mathcal{M}$. For instance, the two vectors $u=(1,2,2,5)$ and $v=(2,1,1,6)$ with $Au=Av=(10,21)$ are connected by the sequence
$$
\begin{psmallmatrix*}[r]1\\2\\2\\5
\end{psmallmatrix*} -
\begin{psmallmatrix*}[r]1\\-1\\-1\\1
\end{psmallmatrix*} =
\begin{psmallmatrix*}[r]0\\3\\3\\4
\end{psmallmatrix*}
\rightarrow
\begin{psmallmatrix*}[r]0\\3\\3\\4
\end{psmallmatrix*} +
\begin{psmallmatrix*}[r]1\\-2\\1\\0
\end{psmallmatrix*} =
\begin{psmallmatrix*}[r]1\\1\\4\\4
\end{psmallmatrix*}
\rightarrow
\begin{psmallmatrix*}[r]1\\1\\4\\4
\end{psmallmatrix*} +
\begin{psmallmatrix*}[r]1\\-1\\-1\\1
\end{psmallmatrix*} =
\begin{psmallmatrix*}[r]2\\0\\3\\5
\end{psmallmatrix*}
\rightarrow
\begin{psmallmatrix*}[r]2\\0\\3\\5
\end{psmallmatrix*} +
\begin{psmallmatrix*}[r]0\\1\\-2\\1
\end{psmallmatrix*} =
\begin{psmallmatrix*}[r]2\\1\\1\\6
\end{psmallmatrix*}.
$$

\end{example}

\subsection{Markov bases for constrained joint degree matrices}
\noindent We now show that the set of all JDMs with a given degree and curvature sequence is the fiber for some matrix $A_\Delta$. This means that the associated Markov basis $\mathcal{M}_{\Delta}$ can be used to explore this fiber. We stress that this matrix and its Markov basis only depend on the maximum degree $\Delta$, and not on the specific degree and curvature frequencies.

We start by encoding a JDM and its constraints as a vector $\gamma\in \mathbb{N}^n$, of length $n=2\binom{\Delta+1}{2}$. This vector $\gamma$ contains the upper-triangular JDM entries $(J_{ab})_{a\leq b}$ and the \emph{slack variables} $(S_{ab})_{a\leq b}$ which we define as 
$$
S_{aa} = \binom{\vert V_a\vert}{2}-J_{aa} \quad\text{~and~}\quad S_{ab} = \vert V_a\vert\cdot\vert V_b\vert - J_{ab}\quad\text{~for $1\leq a<b\leq \Delta$,}
$$
with $\vert V_a\vert$ defined as before from $J$. The slack variables are used to encode the inequalities on the JDM entries given in Theorem \ref{thm: characterization of JDMs}: a matrix $J$ is a valid JDM if and only if the associated slack variables are nonnegative integers, and thus %\textcolor{red}{if and only if} 
$\gamma\in\mathbb{N}^n$. The matrix $J$ can be reconstructed from $\gamma$ by forgetting the slack variables, so we can write $J(\gamma)$. For the JDM of the graph in Example \ref{example: degree and curvature frequencies} and Figure \ref{fig: example JDM and sequences}, the vector $\gamma\in\mathbb{N}^{20}$ equals
\begin{align*}
\gamma &= {\footnotesize{({J_{11}},J_{12},J_{13},J_{14},J_{22},J_{23},J_{24},J_{33},J_{34},J_{44}, S_{11},S_{12},S_{13},S_{14},S_{22},S_{23},S_{24},S_{33},S_{34},S_{44})}}
\\
&= (\,0\,,\,0\,,\,1\,,\,1\,,\,1\,,\,1\,,\,1\,,\,1\,,\,2\,,\,0\,,\, 1\,,\,4\,,\,3\,,\,1\,,\,0\,,\,3\,,\,1\,,\,0\,,\,0\,,\,0\,).
\end{align*}
The first 10 entries of this vector contain the upper-triangular elements of the JDM in Figure \ref{fig: example JDM and sequences}.

Next, we want to construct a matrix $A_{\Delta}$ such that two vectors $\gamma,\gamma'\in\mathbb{N}^n$ represent a JDM with the same degree and curvature frequencies if and only if they satisfy $A_{\Delta}\gamma=A_{\Delta}\gamma'$. In other words, the rows of this matrix $A_{\Delta}$ contain the coefficients of the equations in the $J$ variables (i.\,e., entries of $\gamma$) that compute the degree and curvature frequencies, and coefficients of the equations in the $J$ and $S$ variables (i.e., entries of $\gamma$) given by the definition of the slack variables:
%Next, to construct the matrix $A_{\Delta}$, we consider the equations in the $J$ and $S$ variables that are invariant when fixing the degree and curvature frequencies:
\begin{itemize}
    \item{\makebox[8cm]{$\sum_{b}J_{ab}+J_{aa} \,\,= \vert V_a\vert$, for all $a\in\lbrace 1,\dots,\Delta\rbrace$\hfill} (fixed degree frequencies\footnote{To be precise, this equation fixes $a\vert V_a\vert$ from which $\vert V_a\vert$ can be retrieved.}),}
    
    \item{\makebox[8cm]{$\sum_{a\leq b, a+b=\kappa} J_{ab} \,\,= \vert E_\kappa\vert$, for all $\kappa=\lbrace 2,\dots,2\Delta\rbrace$\hfill} (fixed curvature frequencies),}

    \item{\makebox[8cm]{$S_{aa}+J_{aa}\,\,=\binom{\vert V_a\vert}{2}$, for all $a\in\{1,\dots,\Delta\}$\hfill} (definition of slack variables).}

    \item{\makebox[8cm]{$S_{ab}+J_{ab}\,\,=\vert V_a\vert\cdot\vert V_b\vert$, for all $1\leq a< b\leq \Delta$\hfill} (definition of slack variables).}
\end{itemize}
These $d=\binom{\Delta+1}{2}+3\Delta-1$ equations are linear functions of $\gamma$ with integer coefficients and the coefficients of these equations determine the rows of the $d\times n$ matrix $A_{\Delta}$. To illustrate, we give $A_{\Delta}$ explicitly for $\Delta=4$ in Example \ref{example: delta=4} below and for $\Delta=5$ in Appendix \ref{appendix: delta=5}. 
\begin{example}[$\Delta=4$]\label{example: delta=4}
The matrix $A_4$ is a $21\times 20$ matrix. It can be written as the block-matrix
$$
A_{4} = \begin{pmatrix}
B&0\\
I_{10} & I_{10}
\end{pmatrix},
$$
where $I_{10}$ is the $10\times 10$ identity matrix, $0$ is the zero matrix and $B$ is the $11\times 10$ matrix whose rows encode the fixed degree frequencies (first $4$ rows) and fixed curvature frequencies (last $7$ rows) in terms of the $J$ variables:
\[
%\stackText% MUST BE PREVAILING MODE TO GET LABELS IN TEXT
%\setstacktabbedgap{.6em}
\setstacktabbedgap{.4em} %this was .6
\setstackgap{L}{.9em} %this squishes things vertically
\TABstackTextstyle{\scriptsize\color{red}}
\fixTABwidth{T}
\savestack\collabels{\tabbedCenterstack{11 & 12 & 13 & 14 & 22 & 23 & 24 & 33 & 34 & 44}}
\edef\colwidth{\maxTABwd}
\savestack\rowlabels{\tabbedCenterstack[c]{$\vert V_1\vert$ \\ $\vert V_2\vert$ \\ $\vert V_3\vert$ \\ $\vert V_4\vert$ \\ $\kappa=2$ \\ $\kappa=3$ \\ $\kappa=4$ \\ $\kappa=5$ \\ $\kappa=6$ \\ $\kappa=7$ \\ $\kappa=8$}}
\ensurestackMath{
  B = 
  % \rowlabels
  \stackon{%
    \parenMatrixstack{
    \makebox[\colwidth]{$2$}       &  1 &  1 &  1 &    &    &    &    &    &    \\
      &  1 &    &    &  2 &  1 &  1 &    &    &    \\
      &    &  1 &    &    &  1 &    &  2 &  1 &    \\
      &    &    &  1 &    &    &  1 &    &  1 &  2 \\
    1 &    &    &    &    &    &    &    &    &    \\
      &  1 &    &    &    &    &    &    &    &    \\
      &    &  1 &    &  1 &    &    &    &    &    \\
      &    &    &  1 &    &  1 &    &    &    &    \\
      &    &    &    &    &    &  1 &  1 &    &    \\
      &    &    &    &    &    &    &    &  1 &    \\
      &    &    &    &    &    &    &    &    &  1
    }%
  }{\collabels}
  \rowlabels
}
\]
\noindent The top indices show the correspondence between the columns of $A_{4}$ and the entries of the $J$ matrix, while the indices on the right indicate the equation encoded by the row. %For instance, the second row reflects the sum $2J_{22}+J_{12}+J_{23}+J_{24}$; this equation fixes $2\cdot\vert V_1\vert$ and thus the number of nodes with degree $a=2$. The $7$th row reflects the sum $J_{13}+J_{22}$; this equation fixes the number of edges with Forman--Ricci curvature $4-\kappa=4-4=0$. 
\end{example}

The relevance of the matrix $A_\Delta$ for our problem is shown by the following proposition.
\begin{proposition}\label{prop: JDMs are a fiber}
Let $J$ be a JDM with corresponding vector $\gamma$. The fiber $\mathcal{F}_{A_\Delta}(\gamma)$ contains all JDMs with the same degree and curvature frequencies as $J$. 
\end{proposition}
\begin{proof}
Let $J$ be a JDM with corresponding vector $\gamma$. We consider the fiber
$$
\mathcal{F}_{A_\Delta}(\gamma) = \lbrace \gamma'\in\mathbb{N}^n \,:\, A_\Delta \gamma' = A_\Delta \gamma\rbrace.
$$
Let $\gamma'\in\mathcal{F}_{A_\Delta}$ and write $J'_{ij}$ and $S'_{ij}$ for the corresponding entries in this vector. From the rows of $A_{\Delta}$ that encode the fixed degree frequencies, we obtain
\begin{equation}\label{eq: degree sequences equal}
a\vert V'_a\vert:= \sum_{b} J'_{ab} + J'_{aa} = \sum_{b} J_{ab} + J_{aa} = a\vert V_a\vert \text{~for all $a\in\lbrace 1,\dots,\Delta\rbrace$}.
\end{equation}
From the rows of $A_{\Delta}$ corresponding to the off-diagonal slack equations ($1\leq a<b\leq\Delta)$ and the fact that $\gamma'\in\mathbb{N}^n\geq 0$ and thus $S'_{ab}\geq 0$, we know that
$$
S'_{ab} + J'_{ab} = S_{ab} + J_{ab} = \vert V_a\vert\cdot\vert V_b\vert = \vert V'_a\vert\cdot\vert V'_b\vert \Longrightarrow J'_{ab}\leq \vert V'_a\vert\cdot\vert V'_b\vert.
$$
Similarly, from the diagonal slack equations ($1\leq a\leq \Delta$) and $S'_{aa}\geq 0$, we obtain
$$
S'_{aa} + J'_{aa} = S_{aa}+J_{aa} = \binom{\vert V_{a}\vert}{2} = \binom{\vert V'_{a}\vert}{2}\Longrightarrow J'_{aa}\leq \binom{\vert V'_a\vert}{2}.
$$
By Theorem \ref{thm: characterization of JDMs}, this proves that the matrix $J'=J(\gamma')$ is a JDM. The degree frequencies $(\vert V'_a\vert)_{1\leq a\leq \Delta}$ and $(\vert V_a\vert)_{1\leq a\leq \Delta}$ are equal by \eqref{eq: degree sequences equal}, and the rows of $A_{\Delta}$ which encode the fixed curvature frequencies guarantee that
$$
\sum_{\substack{{a+b}=\kappa\\a\leq b}} J'_{ab} = \sum_{\substack{{a+b}=\kappa\\a\leq b}} J_{ab}
$$
and thus that the curvature frequencies of $J$ and $J'$ are equal. This holds for every $\gamma'\in\mathcal{F}_{A_\Delta}(\gamma)$ which proves that every vector in the fiber indeed corresponds to a JDM with degree and curvature frequencies equal to $J$. The same approach shows that every such JDM is also an element of the fiber, which completes the proof.
\end{proof}

Following Proposition \ref{prop: JDMs are a fiber}, we can thus use a Markov basis $\mathcal{M}_{\Delta}$ corresponding to the matrix $A_{\Delta}$ to obtain all JDMs with fixed curvature and degree frequencies. At this point, it has become clear why we need to consider fixed degrees: to encode the inequalities in the characterization of JDM matrices using slack variables as $S_{ab}+J_{ab}$, we need $\binom{\vert V_a\vert}{2}$ and $\vert V_a\vert\cdot\vert V_b\vert$ to be constant.

A first result on the Markov basis $\mathcal{M}_{\Delta}$ follows from the specific structure of $A_{\Delta}$.
\begin{proposition}\label{prop: slack is negative of J}
The entries of a vector $\beta\in\mathcal{M}_{\Delta}$ corresponding to the slack variables are the negative of the entries corresponding to the $J$ variables.
\end{proposition}
\begin{proof}
By construction, the equation $J_{ab}+S_{ab}$ is invariant so any change $J_{ab}+x$ when applying $\beta$ must be compensated by the opposite change $S_{ab}-x$.
\end{proof}

Following Proposition \ref{prop: slack is negative of J}, we can just consider the $J$-entries of vectors in the Markov basis and we will furthermore rearrange these as a symmetric $\Delta\times\Delta$ matrix and write $\delta J\in\mathcal{M}_\Delta$. We furthermore define\footnote{This notation refers to the fact that $\deg(\delta J)$ is the degree of the homogeneous binomial $\mathbf{x}^{\delta J^+}-\mathbf{x}^{\delta J^-}$ (see Appendix \ref{appendix: background for markov bases}). This should not be confused with the degree of a node.} $\deg(\delta J):=\frac{1}{2}\sum_{a\leq b}\vert J_{ab}\vert$. We will see %later
in the next section that this corresponds to the number of edges involved in the graph operation associated with $\delta J$.

\begin{example}[$\Delta=4$, continued]\label{example: delta=4 continued}
We continue with the matrix $A_4$ from Example \ref{example: delta=4}. The matrix $A_{4}$ has rank $19$ (see also Proposition \ref{prop: rank of Adelta}) and a one-dimensional kernel. The corresponding Markov basis $\mathcal{M}_4$ thus consists of a single vector
$$
\mathcal{M}_4=(0,0,1,-1,-1,1,1,-1,0,0)^T,
$$
where we only show the entries corresponding to the $J$ variables, following Proposition \ref{prop: slack is negative of J}. We can also represent this vector as the symmetric matrix
$$
\delta J = \begin{pmatrix*}[r]
0&0&1&-1\\
0&-1&1&1\\
1&1&-1&0\\
-1&1&0&0
\end{pmatrix*}\text{with $\deg(\delta J)=3.$}
$$    
\end{example}

We now give some further results related to the matrix $A_{\Delta}$.
\begin{proposition}\label{prop: rank of Adelta}
$\rnk(A_{\Delta})= \binom{\Delta+1}{2}+3\Delta-3 =d-2$ for all $\Delta\geq 3$.
\end{proposition}
The proof of Proposition \ref{prop: rank of Adelta} is given in Appendix \ref{appendix: proof of rank Adelta}. Understanding the size of the fiber $\vert\mathcal{F}_{A_\Delta}(\gamma)\vert$ is a hard problem in general. However, making use of the special property that $\ker(A_4)$ is one-dimensional, we obtain the following bound on the size of the fiber in this case. 
\begin{proposition}
Let $J$ be a JDM with $\Delta=4$. Then there are at most
$$
\vert\mathcal{F}_{A_4}(\gamma(J))\vert \leq 
1 + \min\lbrace J_{22},J_{14},J_{33}\rbrace + \min\lbrace J_{13},J_{23},J_{24} \rbrace
$$
JDMs, including $J$, with the same degree and curvature frequencies as $J$.
\end{proposition}
\begin{proof}
We recall from Example \ref{example: delta=4} above that the Markov basis for $A_{4}$ contains a single element
$$
\delta J = \begin{psmallmatrix*}[r]
0&0&1&-1\\
0&-1&1&1\\
1&1&-1&0\\
-1&1&0&0
\end{psmallmatrix*}.
$$
This means that every JDM with the same degree and curvature frequencies as $J$ is of the form $J_k = J+k\times \delta J$ with $k\in\mathbb{Z}$. For $J_k$ to be a JDM, it is necessary that $J_k\geq 0$ and for this reason we know that $k\leq \min\lbrace J_{22},J_{14},J_{33}\rbrace$ and $k\geq \min\lbrace J_{13},J_{23},J_{24}\rbrace$. These conditions are necessary but in general not sufficient, since $J_k\geq 0$ alone does not guarantee that the slack variables are also nonnegative. The possible JDMs from the proposition are thus a subset of $J_{k}$ for $k$ within the given bounds, which completes the proof.
\end{proof}

We note that the bound is the same for all elements of the fiber. To illustrate that the bound can be strict, consider the JDM in the running example in Figure \ref{fig: example JDM and sequences}. Here, $J+\delta J$ is a valid JDM but $J-\delta J$ is not and we thus have just two JDMs in the fiber, whereas $1+\min\lbrace J_{22},J_{14},J_{33}\rbrace+\min\lbrace J_{13},J_{23},J_{24}\rbrace=3$.
%%%
%%%
\subsection{Markov moves and main result}
To combine Theorem \ref{th: transpositions connect simple graphs} on transpositions with Proposition \ref{prop: JDMs are a fiber} on JDMs and Markov bases, we make the observation that applying an element of the Markov basis $J\rightarrow J+\delta J$ to a JDM can be interpreted as a graph operation; this is illustrated in Figure \ref{fig: markov move}.
\begin{figure}[h!]
    \centering
    \includegraphics[width=0.8\textwidth]{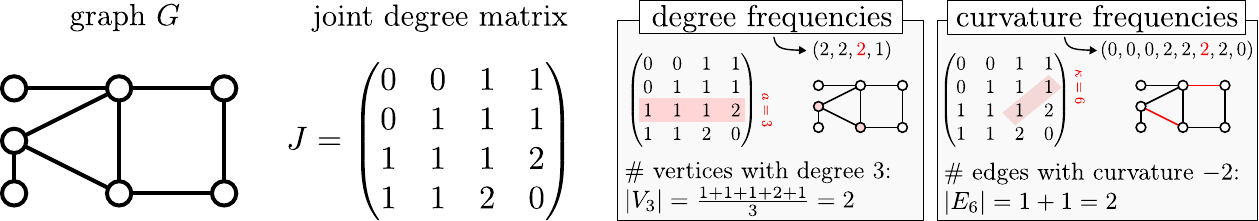}
    \caption{The figure left illustrates how a Markov basis element $\delta J\in\mathcal{M}_\Delta$ can be interpreted as a graph operation that rewires certain edges; we call this a Markov move. The figure right shows a Markov move on the running example.}
    \label{fig: markov move}
\end{figure}

We define a \emph{Markov move} as follows: Let $G$ be a graph and $\delta J$ an element of the Markov basis such that $J(G)+\delta J$ is a JDM. Consider a pair $(\mathcal{E},\pi)$ where $\mathcal{E}=(\lbrace i_\ell,j_\ell\rbrace)_{\ell=1}^{\deg(\delta J)}\subseteq E(G)$ is a subset of edges that has $\vert\delta J_{ab}\vert$ edges with degrees $a,b$ whenever $\delta J_{ab}<0$, and where $\pi$ is a permutation such that the set $\mathcal{E}'=(\lbrace i_\ell,\pi(j_{\ell})\rbrace)_{\ell=1}^{\deg(\delta J)}$ has $\vert \delta J_{ab}\vert$ edges with degrees $a,b$ whenever $\delta J_{ab}>0$. Existence of such a pair follows from the definition of the Markov basis $\delta J$. A Markov move of $G$ is a new graph $G'$ with nodes and edges
$$
V(G')=V(G)\quad\text{~and~}\quad E(G') = (E(G)\backslash \mathcal{E})\cup\mathcal{E}'.
$$
See Figure \ref{fig: markov move} for an example. We now know that all multigraphs with given degree and curvature frequencies are connected by a finite set of graph operations: transpositions, which leave the JDM unchanged, and Markov moves corresponding to the elements of $\mathcal{M}_\Delta$, which change the JDM. This is summarized in the following theorem.
\begin{theorem}\label{th: markov moves and transpositions} The set of multigraphs with given Forman--Ricci curvature and degree frequencies is connected through transpositions and Markov bases. For small $\Delta$, there exist Markov bases of size
%\newpage 
$$
\begin{array}{c||c|c|c|c|c|c|c|c|c}
 \Delta & 
        1 & 2 & 3 & 4 & 5 & 6 & 7 & 8 & \geq 9
        \\
        \hline
        \text{$\vert \mathcal{M}_{A_{\Delta}}\vert$} & 
        0 & 0 & 0 & 1 & 9 & 111 & 2662 & 171964 & ?
\end{array}
$$

\end{theorem}
\begin{proof}
Let $G$ and $G'$ be two multigraphs with the same curvature and degree frequencies If $G$ and $G'$ have the same JDM then Theorem \ref{th: markov moves and transpositions} in the case of multigraphs (see also \cite{stanton_2012_constructing, diaconis_2001_statistical}) says that $G$ can be transformed into $G'$ by transpositions. If $G$ and $G'$ do not have the same JDM, i.\,e., $J\neq J'$, then by Proposition \ref{prop: JDMs are a fiber} and the fundamental theorem of Markov bases, matrix $J$ can be transformed into $J'$ by adding or subtracting a sequence of elements from the Markov basis $\mathcal{M}_{\Delta}$. This implies that $G$ can be transformed using Markov moves into a graph $G''$ such that $J''=J$. Using Theorem \ref{th: markov moves and transpositions} again, $G''$ can then be transformed into $G'$ using transpositions. All intermediate steps in this procedure can be guaranteed to be multigraphs. This concludes the proof of the first part. The results on the size of the Markov bases follows from computations: for $\Delta\in\lbrace 1,2,3\rbrace$ the matrix $A_{\Delta}$ has full rank (see Proposition \ref{prop: rank of Adelta}) and for $\Delta\in\lbrace 4,5,6,7,8\rbrace$ we computed the Markov bases using the software \texttt{4ti2}.
\end{proof}

To illustrate, Figure \ref{fig: markov move and transpositions} shows an example of a sequence of graphs, connected through transpositions and Markov moves; these graphs have the same degree and curvature frequencies. 
\begin{figure}[h!]
    \centering    \includegraphics[width=0.95\textwidth]{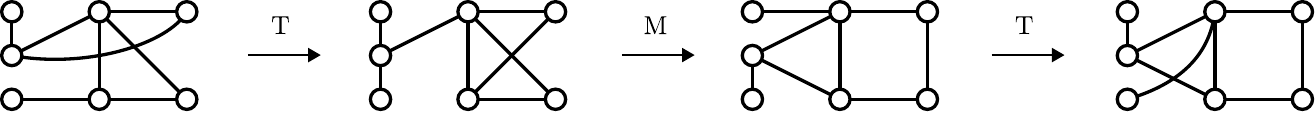}    \caption{A sequence of graphs connected through transpositions (T) and Markov moves (M). Each graph in this sequence has the same degree and curvature frequencies as the graph in Example \ref{example: degree and curvature frequencies}, but not necessarily the same joint degree matrix.}
    \label{fig: markov move and transpositions}
\end{figure}

We note that the set of simple graphs with given curvature and degree frequencies sits inside this bigger set of multigraphs. So in particular, the set of simple graphs is connected by transpositions and Markov moves, with intermediate steps potentially being multigraphs. It is an interesting open question whether all intermediate steps can be simple graphs as well, which would imply a version of Theorem \ref{th: markov moves and transpositions} for simple graphs.

Finally, we note that Theorem \ref{th: markov moves and transpositions} may be used to sample graphs with given degree and curvature frequencies. Algorithm \ref{alg: algorithm} shows one possible approach. \begin{algorithm}[h!]
\caption{Sampling graphs with given curvature and degree frequencies}\label{alg: algorithm}
\begin{algorithmic}
\Require graph $G$, Markov basis $\mathcal{M}_{\Delta}$, number of steps $T$, ratio of Markov moves to transpositions $\Theta\in[0,1]$
\State $t\leftarrow 0$
\State $J\leftarrow$ JDM of $G$
\State $S\leftarrow$ slack variables
\While{$t\leq T$}
    \State $p\leftarrow$ sample random variable in $[0,1]$
    \If{$p\leq \Theta$} \Comment{do a Markov move}
        \State $\delta J\leftarrow$ random element in $\pm\mathcal{M}_{\Delta}$
        \While{$J+\delta J\not\geq 0$ or $S-\delta J\not\geq 0$}
            \State $\delta J\leftarrow$ random element in $\pm\mathcal{M}_{\Delta}$
        \EndWhile
        \State $G'\leftarrow $ Markov move of $G$ associated with $\delta J$
        \State $(G,J,S)\leftarrow(G',J+\delta J,S-\delta J)$
    \ElsIf{$p>\Theta$} \Comment{do a transposition}
        \State $G\leftarrow $ transposition of $G$
    \EndIf
    \State $t\leftarrow t+1$
\EndWhile
\State \textbf{Return:} graph $G'$ 
\end{algorithmic}
\end{algorithm}
The performance of any sampling algorithm based on transpositions and Markov moves will depend on how fast the associated Markov chain mixes. The mixing behavior of transpositions is partially understood via its relation to perfect matching sampling \cite{stanton_2012_constructing}, but understanding and controlling the mixing behavior of the Markov chain associated with general Markov bases is much more challenging \cite{almendra-hernandez_2024_markov, windisch_2016_rapid}.
%%%
%%%
\section{Conclusion and discussion}\label{section: conclusion}
%%%
In this work, we consider the problem of constructing graphs with prescribed discrete curvatures. We take a first step in addressing this methodological gap in network geometry and discrete curvature by approaching the problem in two ways. First, in Section \ref{section: approximate sampling} we develop an evolutionary MCMC-type algorithm to construct graphs whose curvatures approximate a given set. Using this algorithm, we then explore the ensemble of graphs with a fixed discrete curvature sequence. Second, in Section \ref{section: markov bases} we solve the exact reconstruction problem for Forman--Ricci curvature and degree sequences. We show that there exists a finite set of graph rewiring moves --- Markov moves and transpositions --- which connect the set of all graphs with a given Forman--Ricci curvature and degree sequence.

Our work opens several new directions for future research. Most results obtained in this article are limited to specific choices of discrete curvatures, but the same questions are equally relevant for all other notions of discrete curvature. More specifically, our algorithmic approach could be extended by optimizing the expensive computations in the mutation/selection phases for specific curvatures. This would make the sampling algorithm practical for a wider range of curvatures and provide a tool to compare among them. Furthermore, a more careful choice of algorithm design might lead to an understanding of and control over the mixing times and convergence rates of the algorithm. 
Finally, in order to sample graphs which more closely resemble a target graph or which can be used as a null model for real-world networks, our current loss function can be enriched by including additional graph statistics. %We will explore this in future work.
%Furthermore, including other network statistics in the sampling algorithm along with discrete curvature might result in graph ensembles that resemble the target graph more. We will explore this in future work.}

Our results on exact reconstruction for Forman--Ricci target sequences are not as directly extendable to other discrete curvatures. The analysis and solution heavily depends on the simple combinatorial nature of Forman--Ricci curvatures and its connection to joint degrees in a graph. However, an important follow-up question would be to find ways to decrease the prohibitively large size of the Markov bases $\mathcal{M}_{\Delta}$. Adding constraints to the curvature or degree sequences and thus restricting to subspaces of $\ker_{\mathbb{Z}}(A_{\Delta})$ is one possible approach. If the Markov basis size can be controlled, this would enable a practical implementation of the proposed Algorithm \ref{alg: algorithm}.
%%%
%%%
\\~\\
\textbf{Acknowledgements:} The authors thank Florentin M\"unch for a discussion that led to the proof of Proposition \ref{prop: rank of Adelta}.
Michelle Roost was partially supported by a grant from the German-Israeli Foundation for Scientific Research and Development (GIF), grant number 1514.

\appendix
\section{Algebraic background for Markov bases}\label{appendix: background for markov bases}
This brief appendix defines the algebraic concepts that appear in the fundamental theorem of Markov bases, Theorem \ref{thm: fundamental theorem of Markov bases}, and that may not be commonly known. We refer to \cite{MarkovBases, almendra-hernandez_2024_markov} for further background.

Let $K=\mathbb{C}[x_1,\dots,x_n]$ be the ring of polynomials in variables $x_1,\dots,x_n$ and complex coefficients. For a vector $u=(u_1,\dots,u_n)\in\mathbb{N}^n$ we write $\mathbf{x}^u=x_1^{u_1}\cdot\dots\cdot x_n^{u_n}\in K$ for the monomial with exponent vector $u$, and for a vector $\beta\in\mathbb{Z}^n$ we write $\beta^+$ for the vector with entries $(\beta^+)_i = \beta_i$ if $\beta_i\geq 0$ and zero otherwise, and $\beta^-=\beta^+-\beta$. 

An \emph{ideal} in $K$ is a subset $I\subseteq K$ such that $f\in I, g\in K\Rightarrow fg\in I$ and $f,g\in I\Rightarrow f+g\in I$. An ideal $I$ is \emph{generated by} $f_1,\dots,f_k\subseteq I$ if every $g\in I$ can be written as a polynomial combination $g=\sum_{i=1}^k h_i f_i$ for some $h_1,\dots,h_k\in K$.

The fundamental theorem of Markov bases is thus an algebraic statement which says that Markov bases are precisely the generators of some ideal associated with the kernel $\ker_{\mathbb{Z}}(A)$ of an integer matrix $A\in\mathbb{Z}^{d\times n}$; these are called toric ideals.

\section{Example: matrix $A_{5}$ and its Markov basis $\mathcal{M}_5$}\label{appendix: delta=5}
\begin{example}[$\Delta=5$]\label{example: delta=5}
The matrix $A_5$ is a $29\times30$ matrix. It can be written in block-matrix form as
$$
A_{5} = \begin{pmatrix}
B & 0\\
I_{15} & I_{15}
\end{pmatrix},
$$
where $I_{15}$ is the $15\times15$ identity matrix, $0$ is the zero matrix of appropriate dimensions and $B$ is the $14\times 15$ matrix whose rows encode the fixed degree frequencies (first $5$ rows) and fixed curvature frequencies (last $9$ rows) in terms of the $J$ variables:

\[
%\stackText% MUST BE PREVAILING MODE TO GET LABELS IN TEXT
\setstacktabbedgap{.4em} %this was .6
\setstackgap{L}{.9em} %this squishes things vertically
\TABstackTextstyle{\scriptsize\color{red}}
\fixTABwidth{T}
\savestack\collabels{\tabbedCenterstack{11 & 12 & 13 & 14 & 15 & 22 & 23 & 24 & 25 & 33 & 34 & 35 & 44 & 45 & 55}}
\edef\colwidth{\maxTABwd}
\savestack\rowlabels{\tabbedCenterstack[c]{$\vert V_1\vert$ \\ $\vert V_2\vert$ \\ $\vert V_3\vert$ \\ $\vert V_4\vert$ \\ $\vert V_5\vert$ \\ $\kappa=2$ \\ $\kappa=3$ \\ $\kappa=4$ \\ $\kappa=5$ \\ $\kappa=6$ \\ $\kappa=7$ \\ $\kappa=8$ \\ $\kappa=9$ \\ $\kappa=10$}}
\ensurestackMath{
  B = 
  % \rowlabels
  \stackon{%
    \parenMatrixstack{
    \makebox[\colwidth]{$2$}       &  1 &  1 &  1 &  1 &    &    &    &    &    &    &    &    &    &    \\
      &  1 &    &    &    &  2 &  1 &  1 &  1 &    &    &    &    &    &    \\
      &    &  1 &    &    &    &  1 &    &    &  2 &  1 &  1 &    &    &    \\
      &    &    &  1 &    &    &    &    &    &    &  1 &    &  2 &  1 &    \\
      &    &    &    &  1 &    &    &    &  1 &    &    &  1 &    &  1 &  2 \\
    1 &    &    &    &    &    &    &    &    &    &    &    &    &    &    \\
      &  1 &    &    &    &    &    &    &    &    &    &    &    &    &    \\
      &    &  1 &    &    &  1 &    &    &    &    &    &    &    &    &    \\
      &    &    &  1 &    &    &  1 &    &    &    &    &    &    &    &    \\
      &    &    &    &  1 &    &    &  1 &    &  1 &    &    &    &    &    \\
      &    &    &    &    &    &    &    &  1 &    &  1 &    &    &    &    \\
      &    &    &    &    &    &    &    &    &    &    &  1 &  1 &    &    \\
      &    &    &    &    &    &    &    &    &    &    &    &    &  1 &    \\
      &    &    &    &    &    &    &    &    &    &    &    &    &    &  1
    }%
  }{\collabels}
  \rowlabels
}
\]
\noindent For instance, the first row encodes the equation $2 J_{11}+J_{12}+J_{13}+J_{14}+J_{15}$, which keeps the number $\vert V_1\vert$ of nodes with degree $a=1$ fixed. The $8$th row encodes the equation $J_{13}+J_{22}$, which keeps the number of edges with Forman--Ricci curvature $4-\kappa=4-4=0$ fixed.

The matrix $A_5$ has rank 27 (see Proposition \ref{prop: rank of Adelta}) and using \texttt{4ti2} we compute the following Markov basis $\mathcal{M}_{5}$ consisting of $9$ elements:
\begin{align*}
&\delta J_1&=& &\begin{psmallmatrix*}[r]
0&0&0&0&0\\
\phantom{-}0&0&0&1&-1\\
0&0&-1&1&1\\
0&1&1&-1&0\\
0&-1&1&0&0
\end{psmallmatrix*},\, \qquad
&\delta J_2&=& &\begin{psmallmatrix*}[r]
0&0&0&1&-1\\
0&0&-1&0&1\\
0&-1&1&-1&0\\
1&0&-1&0&0\\
-1&1&0&0&0
\end{psmallmatrix*},\, \qquad
&\delta J_3&=& &\begin{psmallmatrix*}[r]
0&0&0&1&-1\\
0&0&-1&1&0\\
0&-1&0&0&1\\
1&1&0&-1&0\\
-1&0&1&0&0
\end{psmallmatrix*},\, \qquad
\\
&\delta J_4&=& &\begin{psmallmatrix*}[r]
0&0&1&-2&1\\
0&-1&2&0&0\\
1&2&-1&0&-1\\
-2&0&0&1&0\\
1&0&-1&0&0
\end{psmallmatrix*},\, \qquad
&\delta J_5&=& &\begin{psmallmatrix*}[r]
0&0&1&-1&0\\
0&-1&1&0&1\\
1&1&0&-1&-1\\
-1&0&-1&1&0\\
0&1&-1&0&0
\end{psmallmatrix*},\, \qquad
&\delta J_6&=& &\begin{psmallmatrix*}[r]
0&0&1&-1&\phantom{-}0\\
0&-1&1&1&0\\
1&1&-1&0&0\\
-1&1&0&0&0\\
0&0&0&0&0
\end{psmallmatrix*},\, \qquad
\\
&\delta J_7&=& &\begin{psmallmatrix*}[r]
0&0&1&0&-1\\
0&-1&0&0&2\\
1&0&1&-2&-1\\
0&0&-2&1&0\\
-1&2&-1&0&0
\end{psmallmatrix*},\, \qquad
&\delta J_8&=& &\begin{psmallmatrix*}[r]
0&0&1&0&-1\\
0&-1&0&1&1\\
1&0&0&-1&0\\
0&1&-1&0&0\\
-1&1&0&0&0
\end{psmallmatrix*},\, \qquad
&\delta J_9&=& &\begin{psmallmatrix*}[r]
0&0&1&0&-1\\
0&-1&0&2&0\\
1&0&-1&0&1\\
0&2&0&-1&0\\
-1&0&1&0&0
\end{psmallmatrix*}. \, \qquad
\end{align*}
%If we define\footnote{This notation refers to the fact that $\deg(\delta J)$ is the degree of the homogeneous binomial $\mathbf{x}^{\delta J^+}-\mathbf{x}^{\delta J^-}$ (see Appendix \ref{appendix: background for markov bases}). This should not be confused with the degree of a node.} $\deg(\delta J):=\frac{1}{2}\sum_{a\leq b}\vert J_{ab}\vert$, we see that 
Five bases have $\deg(\delta J)=3$, two bases have $\deg(\delta J)=4$ and two bases have $\deg(\delta J)=5$. %We will see later that this corresponds to the number of edges involved in the graph operation associated with $\delta J$.
\end{example}

\section{Rank of the matrix $A_{\Delta}$}\label{appendix: proof of rank Adelta}
Here, we prove Proposition \ref{prop: rank of Adelta}, which states that matrix $A_{\Delta}$ has rank $d-2$, where we recall that $d=\binom{\Delta+1}{2}+3\Delta-1$ is the number of rows of $A$.
\\
\begin{proof} (of Proposition \ref{prop: rank of Adelta})
We first give two independent linear dependencies between the rows of $A_{\Delta}$ to show that $\rnk(A_{\Delta})\leq d-2$. The first linear relation between the rows is
$$
\sum_{a=1}^\Delta(\text{row $\vert V_a\vert$}) - 2\sum_{\kappa=2}^{2\Delta}(\text{row $\kappa$}) = 
\sum_{a=1}^{\Delta} \Biggl(\sum_c J_{ac}+J_{aa}\Biggr)-2\sum_{\kappa=2}^{2\Delta}\Biggl( \sum_{\substack{a+b=\kappa\\a\leq b}}J_{ab}\Biggr) = 0.
$$
The second linear relation between the rows is
\begin{align*}
\sum_{a=1}^\Delta a\times (\text{row $\vert V_a\vert$}) - \sum_{\kappa=2}^{2\Delta}\kappa\times(\text{row $\kappa$}) =\sum_{a=1}^{\Delta}a\left(\sum_{c}J_{ac}+J_{aa}\right) - \sum_{\kappa=2}^{2\Delta}\kappa\Biggl(\sum_{\substack{a+b=\kappa\\a\leq b}}J_{ab}\Biggr)=0.
\end{align*}
To prove that the latter relation holds, we simplify the first sum to $\sum_{v\in V}\deg(v)^2$ and the second sum to $\sum_{e\in E}\kappa(e)$, where $\kappa(\lbrace u,v\rbrace)=\deg(u)+\deg(v)$. We further simplify
\begin{align*}
\sum_{e\in E}\kappa(e) = \sum_{\lbrace u,v\rbrace\in E}(\deg(u)+\deg(v)) =  \sum_{v\in V}\sum_{u\sim v}\deg(v) =  \sum_{v\in V}\deg(v)^2,
\end{align*}
which confirms that the second linear relation equals zero. Next, we show that the matrix $A_{\Delta}$ has a nonsingular $(d-2)$-minor $\tilde{A}$. This implies that $\rnk(A_{\Delta})\geq d-2$ which then completes the proof.

First, note that $A_{\Delta}$ can be written as $A_{\Delta}=\left(\begin{smallmatrix}B&0\\I & I\end{smallmatrix}\right)$, where $0$ is the zero matrix, $I$ is the identity matrix of size $\binom{\Delta+1}{2}$ and $B$ a $(\Delta-1)\times\binom{\Delta+1}{2}$ matrix. The rows corresponding to the identity matrices encode the fixed slack variables and the rows corresponding to $B$ encode the fixed degree and curvature frequencies; as a result, the columns of $B$ correspond to $J$ variables. See Examples \ref{example: delta=4} and \ref{example: delta=5} for a concrete example. We show that $B$ has a nonsingular $(3\Delta-3)$-minor $\tilde{B}$ which implies that $A$ has the required nonsingular $(d-2)$-minor $\tilde{A}=\left(\begin{smallmatrix}\tilde{B}&0\\I&I\end{smallmatrix}\right)$.

Consider the following submatrices of $B$:
\[
%\stackText% MUST BE PREVAILING MODE TO GET LABELS IN TEXT
\setstacktabbedgap{-0.2em} %this was .6
\setstackgap{L}{.9em} %this squishes things vertically
\TABstackTextstyle{\scriptsize\color{red}}
\fixTABwidth{T}
\savestack\collabels{\tabbedCenterstack{$\Delta\Delta$ & $(\Delta-1)\Delta$ & $(\Delta-1)(\Delta-1)$ }}
\edef\colwidth{\maxTABwd}
\savestack\rowlabels{\tabbedCenterstack[l]{$\kappa=2\Delta$ \\ $\kappa=2\Delta-1$ \\ $\kappa=2\Delta-2$}}
\ensurestackMath{
  B_1 = 
  % \rowlabels
  \stackon{%
    \parenMatrixstack{
    \makebox[\colwidth]      
    1 & 0 & 0\\
    0 & 1 & 0\\
    0 & 0 & 1
    }%
  }{\collabels}
  \rowlabels
}
\]
\[
%\stackText% MUST BE PREVAILING MODE TO GET LABELS IN TEXT
\setstacktabbedgap{0.4em} %this was .6
\setstackgap{L}{.9em} %this squishes things vertically
\TABstackTextstyle{\scriptsize\color{red}}
\fixTABwidth{T}
\savestack\collabels{\tabbedCenterstack{$(\Delta-i)(\Delta+2-i)$ & $(\Delta-i)(\Delta+1-i)$ & $(\Delta-i)(\Delta-i)$ }}
\edef\colwidth{\maxTABwd}
% \savestack\rowlabels{\tabbedCenterstack[l]{$\vert V_{\Delta-i}\vert$ \\ $\kappa=2\Delta-1-i$ \\ $\kappa=2\Delta-2-i$}}
\savestack\rowlabels{\tabbedCenterstack[l]{$\vert V_{\Delta-i}\vert$ \\ $\kappa=2\Delta-2i + 1$ \\ $\kappa=2\Delta-2i$}}
\ensurestackMath{
  B^i = 
  % \rowlabels
  \stackon{%
    \parenMatrixstack{
    \makebox[\colwidth]      
    1 & 1 & 2\\
    0 & 1 & 0\\
    0 & 0 & 1
    }%
  }{\collabels}
  \rowlabels
}
\]
for $i\in\lbrace 2,\dots, \Delta-1\rbrace$. After reordering rows and columns, we can write matrix $\tilde{B}$ as
\[
    \tilde{B}=\left(
    \begin{array}{ccccc}
    B_1                                    \\
      & B^2             &   & \text{\huge$\star$}\\
      &               & \ddots                \\
      & \text{\huge0} &   & B^{(\Delta-1)}
    \end{array}
    \right),
\]  
where $\star$ indicates entries that are not relevant. This is by construction a ($3\Delta-3$)-minor of $B$. It is an upper-triangular matrix with ones along the diagonal, which means that it is nonsingular. As noted, this implies that the $(d-2)$-minor $\tilde{A}$ is nonsingular, which completes the proof.
\end{proof}

\newpage
\bibliographystyle{abbrv}
\bibliography{bibliography.bib}
\end{document}